\documentclass[11pt]{article}

\usepackage{amssymb,amsthm,amsmath}
\usepackage{commath,hyperref}
\usepackage{enumerate,xcolor}
\usepackage{natbib}
\usepackage{graphicx}
\usepackage{caption,setspace}
\usepackage{subcaption}

\newcommand{\ds}{\displaystyle}
\newcommand{\E}{\mathbb{E}}
\newcommand{\Var}{\text{Var}}
\newcommand{\Cov}{\text{Cov}}
\newcommand{\X}{\mathcal{X}}

\newcommand\numberthis{\addtocounter{equation}{1}\tag{\theequation}}

\usepackage{fullpage}
\usepackage{comment}
\usepackage[linesnumbered, ruled]{algorithm2e}
\allowdisplaybreaks
\newtheorem{theorem}{Theorem}
\newtheorem{lemma}{Lemma}

\theoremstyle{remark}

\newtheorem{ass}{Assumption}
\newtheorem{remark}{Remark}
\title{Globally-centered autocovariances in MCMC}

\author{
  Medha Agarwal \\
  Dept. of Mathematics and Statistics\\
  IIT Kanpur\\
  \texttt{medhaaga@iitk.ac.in} \and
 Dootika Vats\footnote{Corresponding author}\\
  Dept.of Mathematics and Statistics\\
  IIT Kanpur\\
  \texttt{dootika@iitk.ac.in}
}

\begin{document}

\maketitle

\doublespacing
\begin{abstract}
Autocovariances are a fundamental quantity of interest in Markov chain Monte Carlo (MCMC) simulations with autocorrelation function (ACF) plots being an integral visualization tool for performance assessment. Unfortunately, for slow-mixing Markov chains, the empirical autocovariance can highly underestimate the truth. For multiple-chain MCMC sampling,  we propose a globally-centered estimator of the autocovariance function (G-ACvF) that exhibits significant theoretical and empirical improvements. We show that the bias of the G-ACvF estimator is smaller than the bias of the current state-of-the-art.  The impact of this improved estimator is evident in three critical output analysis applications: (1) ACF plots, (2) estimates of the Monte Carlo asymptotic covariance matrix, and (3) estimates of the effective sample size. Under weak conditions, we establish strong consistency of our improved asymptotic covariance estimator, and obtain its large-sample bias and variance.  The performance of the new estimators is demonstrated through various examples.
\end{abstract}

\section{Introduction} \label{sec:intro}

Advancements in modern personal computing have made it easy to run parallel Markov chain Monte Carlo (MCMC) implementations. This is particularly useful for slow-mixing Markov chains where the starting points of the chains are spread over the state-space in order to more accurately capture characteristics of the target distribution. The output from $m$ parallel chains is then summarized, visually and quantitatively, to assess the empirical mixing properties of the chains and the quality of Monte Carlo estimators.

A key quantity that drives MCMC output analysis is the autocovariance function (ACvF).
 Estimators of ACvF are only available for single-chain implementations, with a parallel-chain version obtained by naive averaging. As we will demonstrate, this defeats the purpose of running parallel Markov chains from dispersed starting values.

Let $F$ be the target distribution with mean $\mu$, defined on $\X \subseteq \mathbb{R}^d$, equipped with a countably generated $\sigma$-field, $\mathcal{B}(\X)$. For $s = 1, \dots, m$, let $\{X_{s,t}\}_{t\geq1}$ be the $s${th} Harris ergodic $F$-invariant Markov chain \citep[see][for definitions]{meyn:twee:2009} employed to learn characteristics about $F$. The process is covariance stationary with the lag $k$ ACvF defined as 
\[
    \Gamma(k) = \Cov_F(X_{s,1}, X_{s,1+k})= \mathbb{E}_F \left[(X_{s,1} - \mu)(X_{s,1+k} - \mu)^T \right]\,.
\]
Estimating $\Gamma(k)$ is critical to assessing the quality of the sampler and the reliability of Monte Carlo estimators. Let $\bar{X}_s = n^{-1} \sum_{t=1}^{n} X_{s,t}$ denote the Monte Carlo estimator of $\mu$ from the $s$th chain. The standard estimator for $\Gamma(k)$ is the sample autocovariance matrix at lag $k \geq 0$:
\begin{equation} \label{eq:empirical_ACvF}
    \hat{\Gamma}_s(k) = \dfrac{1}{n}\sum_{t=1}^{n-k} \left(X_{s,t} - \bar{X}_s \right) \left(X_{s, t + k} - \bar{X}_s \right)^T\,,
\end{equation}
and for $k < 0$, $\hat{\Gamma}_s(k) = \hat{\Gamma}_s(-k)^T$. For a single-chain MCMC run, the estimator $\hat{\Gamma}_s(k)$ is used to construct ACF plots, to estimate the long-run variance of Monte Carlo estimators \citep{hannan:1970,dame:1991}, and to estimate effective sample size (ESS) \citep{kass:carlin:gelman:neal:1998,gong:fleg:2016,vats:fleg:jon:2019}. However, there is no unified approach to constructing estimators of $\Gamma(k)$ for parallel-chain implementations. Slow-mixing chains take time to traverse the space so that $\bar{X}_s$ over  all $s$, can be vastly different. Consequently, for Markov chains with positive autocorrelations, $\hat{\Gamma}_s(k)$ typically underestimates $\Gamma(k)$, leading to a false sense of security about the quality of the process.

 We propose a globally-centered estimator of ACvF (G-ACvF) that centers the Markov chains around the global mean from all $m$ chains. We show that the bias for G-ACvF is lower than $\hat{\Gamma}_s(k)$, and through various examples, demonstrate improved estimation. We employ the G-ACvF estimators to construct ACF plots and a demonstrative example is at the end of this section.

 
Estimators of ACvFs are used to estimate the long-run variance of Monte Carlo averages. Specifically, spectral variance (SV) estimators  are used to estimate $\Sigma = \sum_{k=-\infty}^{\infty} \Gamma(k)$ \citep{andr:1991,dame:1991,fleg:jone:2010}. We replace $\hat{\Gamma}_s(k)$ with G-ACvF in SV estimators to obtain a globally-centered SV (G-SV) estimator and demonstrate strong consistency under weak conditions. In the spirit of \cite{andr:1991}, we also obtain large-sample bias and variance of the resulting estimator. SV estimators can be prohibitively slow to  compute \citep{liu:fleg:2018}. To relieve the computational burden, we adapt the fast algorithm of \cite{heberle2017fast} to dramatically reduce computation time. The G-SV estimator  is employed in the computation of ESS. We will show that using the G-SV estimator for estimating $\Sigma$ safeguards users against early termination of the MCMC process.

\subsection{Demonstrative example} 
\label{sub:demonstrative_example}


We use ACF plots to demonstrate the striking difference in the estimation of $\Gamma(k)$. Consider a random-walk Metropolis-Hastings sampler for a univariate mixture of Gaussians. Let
\[
f(x) = 0.7\,f(x; -5, 1) + 0.3\,f(x; 5, 0.5)\,,
\]
be the target density where $f(x; a,b)$ is the density of a normal distribution with mean $a$ and variance $b$. 
We set $m = 2$ with starting values distributed to the two modes. The trace plots in Figure~\ref{fig:gaussian-trace} indicate that in the first $10^4$ steps the chains do not jump modes so that both Markov chains yield significantly different estimates of the population mean, $\mu$. At $10^5$ sample size, both Markov chains have traversed the state space reasonably and yield similar estimates of $\mu$. Further, in Figure~\ref{fig:gaussian-trace}, for $n = 10^4$ we present the ACF plots using $\hat{\Gamma}_s(k)$ and our proposed G-ACvF estimator. The blue curves are the respective estimates at $n = 10^5$. At $n = 10^5$ when the chains have similar means, the G-ACF and locally-centered ACF are equivalent.  However, for $n = 10^4$, $\hat{\Gamma}_s(k)$ critically underestimates the correlation, producing a misleading visual of the quality of the Markov chain.  G-ACF  accounts for the discrepancy in sample means between the two chains, leading to a far improved quality of estimation.

\begin{figure}[htbp]
\centering
   \includegraphics[width=.62\linewidth]{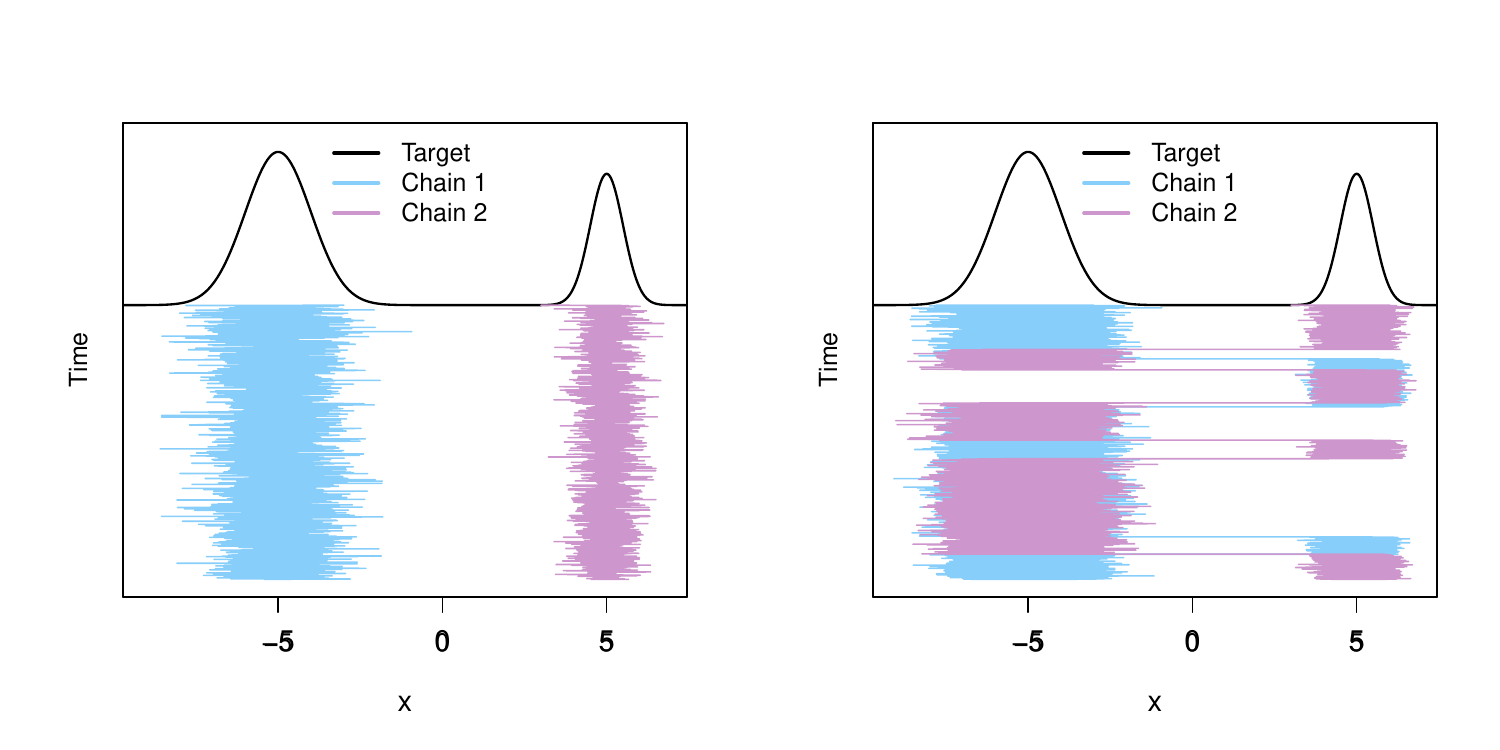}\\
\vspace{-.5cm}
    \includegraphics[width=.62\textwidth]{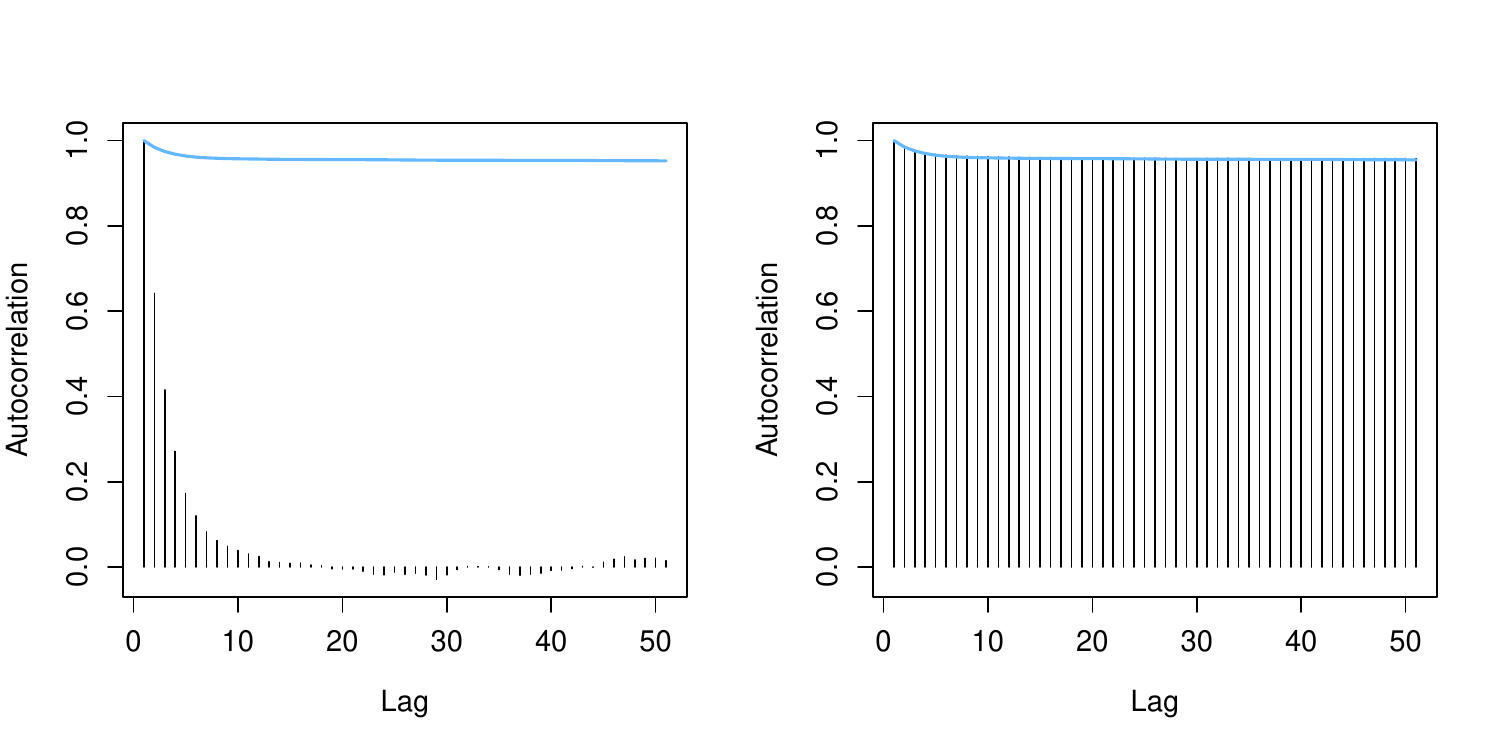} 
    \caption{Top: Target density and trace plots for two chains at $n = 10^4$ (left) and $n = 10^5$ (right). Bottom: ACF plots for the first chain using local-centering (left) and global-centering (right). Histogram estimates at $n = 10^4$ with the blue curve being estimates from $n = 10^5$.}
    \label{fig:gaussian-trace}
\end{figure}



\section{Globally-centered autocovariance} \label{sec:G-ACF}

Let $P$ denote the Markov transition kernel that uniquely determines the ACvF under stationarity. Recall that the sample mean of the $s${th} Markov chain is $\bar{X}_s$ and denote the global mean vector by $\bar{\bar{X}} = m^{-1}\sum_{s = 1}^{m}\bar{X}_s$. For $k \geq 0$, we define the globally-centered ACvF (G-ACvF) estimator for the $s${th} Markov chain as
\begin{equation} \label{eq:G-ACvF}
    \hat{\Gamma}_{G,s}(k) = \dfrac{1}{n} \sum_{t=1}^{n-k}(X_{s,t}-\bar{\bar{X}})(X_{s,t+k}-\bar{\bar{X}})^T,
\end{equation}
with $\hat{\Gamma}_{G,s}(k) = \hat{\Gamma}_{G,s}(-k)^T$ for $k < 0$. In the event that all $m$ Markov chains have been run long enough, $\bar{X}_s \approx \bar{\bar{X}}$, and hence  $\hat{\Gamma}_{s} (k) \approx \hat{\Gamma}_{G,s}(k)$. However, for shorter runs or for slow-mixing chains, $\Gamma(k)$ is more appropriately estimated by $\hat{\Gamma}_{G,s}$ as it utilizes information from all chains and accounts for disparity between estimates of $\mu$. For $q \geq 1$, let 
\begin{equation*}
\Phi^{(q)} = \sum_{ k= -\infty}^{\infty}\abs{k}^q \Gamma(k)\,,  
\end{equation*}
and let $\Phi^{(1)}$ be denoted by $\Phi$. Let $\|\cdot\|$ denote Euclidean norm. The proof of the following theorem is in the supplement and is similar to that of $\hat{\Gamma}_{s}(k)$ \citep{priestley1981spectral}.

\begin{theorem} \label{th:G-ACF_bias} Let $\E_F \|X_{1,1}\|^{2 + \delta} < \infty$ for some $\delta > 0$. If $P$ is polynomially ergodic of order $\xi > (2 + \epsilon)/(1 + 2/\delta)$ for some $\epsilon > 0$,  then,
\[
   \mathbb{E}_F\left[\hat{\Gamma}_{G,s}(k) \right] = \left(1- \dfrac{|k|}{n}\right) \left(\Gamma(k) - \dfrac{\Sigma}{mn} - \dfrac{\Phi}{mn^2}\right)  + o \left(n^{-1} \right)\,.
\]
\end{theorem}

\begin{remark}
For any matrix $M$, let $M^{ii}$ denote the $i$th diagonal element. The proof of Theorem~\ref{th:G-ACF_bias} yields the following intermediary expression (up to a $o(n^{-2})$) term:
\begin{align*} \label{eq:G-ACvF_bias}
\mathbb{E}\left[\hat{\Gamma}_{G,s}(k)^{ii}- \hat{\Gamma}_s(k)^{ii} \right] &=
     \left(1-m^{-1} \right)\left[\dfrac{\Sigma^{ii}}{n} + \dfrac{\abs{k}}{n^2}\left(\Sigma^{ii} - 2\sum_{h=0}^{n-1}\Gamma(h)^{ii}\right) + \left(1 + \dfrac{\abs{k}}{n}\right)\dfrac{\Phi^{ii}}{n^2}\right] \,.
\end{align*}
In the presence of positive autocorrelation, $\Phi^{ii}$ is positive yielding a positive difference in expectation. Since both estimators are under-biased for $\Gamma(k)^{ii}$ (as evidenced from Theorem~\ref{th:G-ACF_bias}), this implies the G-ACvF estimator yields bias reduction.
\end{remark}
\begin{remark}
Polynomial ergodicity and the moment conditions are required to ensure $\Phi$ and $\Sigma$ are finite. Theorem~\ref{th:G-ACF_bias} can be stated more generally for $\alpha$-mixing processes, but we limit our attention to Markov chains. 
\end{remark}

For component $i$, the autocorrelation is defined as
\[
\rho^{(i)}(k) = \dfrac{\Gamma^{ii}(k)}{\Gamma^{ii}(0)}\,,
\]
and is instrumental in visualizing the serial correlation in the components of the Markov chain. A standard estimator of the autocorrelation is constructed from $\hat{\Gamma}_s(k)$. Instead, we advocate for using G-ACvF, so that,
\begin{equation}
\label{eq:acf}
 \hat{\rho}_{G,s}^{(i)}(k) = \dfrac{ \hat{\Gamma}^{ii}_{G,s} (k)}{\hat{\Gamma}^{ii}_{G,s} (0)}. 
\end{equation}
The globally-centered autocorrelation provides a far more realistic assessment of the correlation structure of the marginal components of the chain  as evidenced in Figure~\ref{fig:gaussian-trace}. 
%
%
We end this section with the average G-ACvF and G-ACF over all $m$ chains, that provides a measure of the overall correlation structure induced by the Markov transition $P$. 
\[
\hat{\Gamma}_G(k) = \dfrac{1}{m}\sum_{s=1}^m \hat{\Gamma}_{G,s}(k) \quad \text{ and } \quad \hat{\rho}^{(i)}_G(k) = \dfrac{1}{m}\sum_{s=1}^m \hat{\rho}^{(i)}_{G,s}(k)\,.
\]

\section{Long-run variance estimators} \label{sec:variance_est}

A critical need of autocovariances is in the assessment of Monte Carlo variability of estimators.  Let $g:\X \to \mathbb{R}^p$ be an $F$-integrable function so that interest is in estimating
$\mu_g = \mathbb{E}_F[g(X)]$.
 %
Set $\{Y_{s,t}\}_{t \geq 1} = \{g(X_{s,t})\}_{t \geq 1}$ for $s = 1, \dots, m$. Let $\bar{Y}_s = n^{-1}\sum_{t=1}^{n}Y_{s,t}$  and  $\bar{\bar{Y}} = m^{-1}\sum_{s=1}^{m}\bar{Y}_s$. By Birkhoff's ergodic theorem,  $\bar{\bar{Y}} \to \mu_g$ with probability 1 as $n \to \infty$. An asymptotic sampling distribution may be available via a Markov chain central limit theorem (CLT) if there exists a $p \times p$ positive-definite matrix $\Sigma$ such that
\begin{align*}
  \sqrt{mn}(\bar{\bar{Y}} - \mu_g) \xrightarrow{d} N(0,\Sigma)\, \quad \text{ where }  \quad 
  \Sigma = \sum_{k = -\infty}^{\infty} \Cov_F \left( Y_{1,1}, Y_{1,1+k} \right) := \sum_{k = -\infty}^{\infty}\Upsilon (k)\,.
\end{align*}
The goal in output analysis for MCMC is to estimate $\Sigma$ in order to assess variability in $\bar{\bar{Y}}$ \citep{fleg:hara:jone:2008,roy:2019,vats:rob:fle:jon:2020}. There is a rich literature on estimating $\Sigma$ for single-chain MCMC implementations. The most common are SV estimators \citep{andr:1991,vats:fleg:jon:2018} and batch means estimators \citep{chen:seila:1987,vats:fleg:jon:2019}. Recently, \cite{gupta:vats:2020} constructed a replicated batch means estimator for estimating $\Sigma$ from parallel Markov chains. Batch means estimators are computationally more efficient than SV estimators, whereas SV estimators are more reliable \citep{damerdji:1995,fleg:jone:2010}. Here, we utilize G-ACvF estimators to construct globally-centered SV (G-SV) estimator of $\Sigma$. Using the method of \cite{heberle2017fast}, we also provide a computationally efficient implementation of the G-SV estimator. 

For $k \geq 0$, the locally and globally-centered  estimators of $\Upsilon(k)$ are
\[
\hat{\Upsilon}_s(k) = \dfrac{1}{n} \ds \sum_{t=1}^{n - k} (Y_{s,t} - \bar{Y}_s)(Y_{s,t+k} - \bar{Y}_s)^T \quad \text{ and } \quad \hat{\Upsilon}_{G,s}(k) = \dfrac{1}{n} \ds \sum_{t=1}^{n - k} (Y_{s,t} - \bar{\bar{Y}})(Y_{s,t+k} - \bar{\bar{Y}})^T\,,
\] 
respectively. Let $\hat{\Upsilon}_G(k) = m^{-1} \sum_{k=1}^{m}\hat{\Upsilon}_{G,s}(k)$. SV estimators are weighted and truncated sums of estimated ACvFs. For some $c \geq 1$, let $w: \mathbb{R} \to [-c,c]$ be a lag window function and $b_n \in \mathbb{N}$ be a truncation point.
\begin{ass}
\label{ass:lag_window}
The lag window $w(x)$ is continuous at all but a finite number of points, is a bounded and even function with $w(0)=1$, \; $\int_{-\infty}^{\infty}w^2(x)dx < \infty$, and $\int_{-\infty}^{\infty} \abs{w(x)} < \infty$.
\end{ass}

Assumption~\ref{ass:lag_window} is standard \citep[see][]{ande:1971}; we employ the popular Bartlett lag window in our simulations for which $w(x) = 1-|x|$ for $|x| \leq 1$ and 0 otherwise. The (locally-centered) SV estimator of $\Sigma$ for chain $s$ is
\begin{equation} \label{eq:sve}
    \hat{\Sigma}_{s} = \sum_{k=-b_n+1}^{b_n-1}w\left(\dfrac{k}{b_n}\right)\hat{\Upsilon}_s(k)\,.
\end{equation}

Large-sample properties of $\hat{\Sigma}_s$ have been widely studied. \cite{vats:fleg:jon:2018} provide conditions for strong consistency while \cite{fleg:jone:2010,hannan:1970} obtain bias and variance. These results extend naturally to an average SV (ASV) estimator, $\hat{\Sigma}_A := m^{-1} \sum_{s=1}^{m}\hat{\Sigma}_{s}$.


\subsection{Globally-centered spectral variance estimators} 
\label{sub:globally_centered_spectral_variance_estimators}


We define the G-SV estimator as the weighted and truncated sum of G-ACvFs
\begin{equation}
\label{eq:gsve_estimator}
    \hat{\Sigma}_{G} = \sum_{k= -b_n+1}^{b_n-1}w\left(\dfrac{k}{b_n}\right)\hat{\Upsilon}_{G}(k)\,.
\end{equation}
%
%



\subsubsection{Theoretical results} \label{sec:G-SVE}

First, we provide conditions for strong consistency. Strong consistency is particularly important to ensure that sequential stopping rules in MCMC yield correct coverage at termination \citep{fleg:gong:2015,glyn:whit:1992,vats:fleg:jon:there-yet}. A critical assumption is that of a strong invariance principle which the following theorem establishes. Let $B(n)$ denotes a standard $p$-dimensional Brownian motion. The proof of Theorem~\ref{th:consistency} is in the supplement.

\begin{theorem}[\cite{kuel:phil:1980,vats:fleg:jon:2018}]
  \label{thm:kuelbs}
Let $\E_F\|Y_{1,1}\|^{2+ \delta} < \infty$ for $\delta > 0$ and let $P$ be polynomially ergodic of order $\xi > (q + 1 + \epsilon)/(1 + 2/\delta)$ for $q \geq 1$.  There exists a $p \times p$ lower triangular matrix $L$ with $LL^T = \Sigma$, a non-negative function $\psi(n) = n^{1/2 - \lambda}$ for some $\lambda > 0$, a finite random variable $D$, and a sufficiently rich probability space $\Omega$ such that for all $n > n_0$,
\[
\left\|\sum_{t=1}^{n}Y_t - n\mu_g - LB(n)\right\| < D\psi(n) \qquad \text{  with probability 1}\,.
\]
\end{theorem}

\begin{theorem}
\label{th:consistency}
 Let the assumptions of Theorem~\ref{thm:kuelbs} hold with $q = 1$. If $\hat{\Sigma}_{s} \xrightarrow{a.s.} \Sigma$ for all $s$, and $n^{-1}{b_n \log \log n} \to 0 \textrm{ as } n \to \infty$, then $\hat{\Sigma}_{G} \overset{a.s.}{\to} \Sigma$ as $n \to \infty$.
\end{theorem} 

Conditions for strong consistency of $\hat{\Sigma}_s$ for polynomially ergodic Markov chains are in \cite{vats:fleg:jon:2018}. Typically, $b_n = \lfloor n^{\nu} \rfloor$ for some $0< \nu <1$ for which $n^{-1} b_n \log \log n \to 0$, thus Theorem~\ref{th:consistency} presents no added conditions for strong consistency. Our next two results establish large-sample bias and variance for G-SV and mimic those of $\hat{\Sigma}_s$ \citep{hannan:1970}. Let $\Sigma^{ij}$ and $\hat{\Sigma}_G^{ij}$ denote the $ij$th element of the matrix $\Sigma$ and $\hat{\Sigma}_G$, respectively. The proofs of the results below can be found the supplement.

\begin{theorem}\label{th:G-SVE_bias}
Let the assumptions of Theorem~\ref{thm:kuelbs} hold with $q$ such that 
\[
\lim _{x \to 0} \dfrac{1 - w(x)}{\abs{x}^q} = k_q < \infty\,
\] 
and $b_n^{q+1}/n \to 0$ as $n \to \infty$. Then, $ \lim_{n \to \infty}b_n^q\mathbb{E} [\hat{\Sigma}_{G} - \Sigma ] = -k_q\Phi^{(q)}\,.$
\end{theorem}
\cite{hannan:1970} proved a similar bias result assuming $\mu_g$ was known with the rate $b_n^q/n \to 0$ as $n \to \infty$. Similar to \cite{ande:1971} for the univariate case, if $\mu_g$ is replaced by $\bar{\bar{Y}}$, we require $b_n^{q+1}/n \to 0$ as $n\to \infty$ for the multivariate case.

\begin{theorem} \label{th:G-SVE_variance}
 Let the assumptions of Theorem~\ref{thm:kuelbs} hold and let $\E[D^4] < \infty$  and  $\E_F \|Y_{1,1}\|^4 < \infty$, then $ \lim_{n \to \infty} b_n^{-1}{n}\Var \left(\hat{\Sigma}_{G}^{ij} \right) = [\Sigma_{ii}\Sigma_{jj} + \Sigma_{ij}^2]\int_{-\infty}^{\infty}w(x)^2dx $.
\end{theorem}
For the Bartlett lag window, $q = 1$ with $k_q = 1$ and $\int w(x)^2 dx = 2/3$. For other popular lag windows, see \cite{ande:1971}.

\begin{remark}
  The asymptotic results in Theorems~\ref{th:G-SVE_bias} and \ref{th:G-SVE_variance} are similar to that of the A-SV estimator. This is unsurprising since the global and the local means are asymptotically equivalent. Unfortunately, finite sample results are unavailable for even the A-SV estimator.
\end{remark}
\subsubsection{Fast implementation} 
\label{ssub:fast_implementation}

The SV estimator, despite having good statistical properties, poses limitations due to slow computation. The complexity of the SV estimator is $\mathcal{O}(b_n n p^2)$. For slow-mixing Markov chains, $n$ and $b_n$ can be prohibitively large, limiting the use of SV estimators.

We adapt the  fast Fourier transform based algorithm of \cite{heberle2017fast} to calculate the G-SV estimator.  Let $w_k = w(k/b_n)$ and let $T(w)$ be the $n \times n$ Toeplitz matrix with the first column being $(1 ~ w_1 ~ w_2 ~ \dots ~ w_{n-1})^T$. Notice an alternate formulation of $\hat{\Sigma}_s$
\begin{equation*} \label{eq:kyriakoulis}
    \hat{\Sigma}_s = \dfrac{1}{n}A_s^T T(w) A_s, \qquad \textrm{ where } \quad  A_s = \begin{pmatrix}
    Y_{s,1} - \bar{Y}_s  & \dots & Y_{s,n} - \bar{Y}_s
\end{pmatrix}^T \,.
\end{equation*}
 Let $w^* = (1 ~ w_1 ~ w_2 ~ \dots, ~ w_{n-1}, ~0, ~w_{n-1}, \dots, w_1)^T$ and set $C(w^*)$ to be a symmetric circulant matrix such that the matrix truncation $C(w^*)_{1:n, 1:n} = T(w)$. Let $M_{(j)}$ denote the $j$th column of a matrix $M$ and $v^{(i)}$ denote the $i$th element of a vector $v$. With inputs $C(w^*)$ and $A_s$, Algorithm~\ref{algo:herberle} produces $\hat{\Sigma}_s$ exactly. For more details, see \cite{heberle2017fast}.

\begin{algorithm}[htbp] 
\KwIn{$C(w^*)$ and $A_s$}
\DontPrintSemicolon
\SetAlgoLined
Compute eigenvalues $\lambda_i$ of  $C(w^*)_{(1)}$ using a discrete Fourier transform, $i = 1, \dots, 2n$\;
Construct $2n \times p$ matrix $A^*_s = (A_s^T  \quad 0_{n \times p})^T$\;

\For{$j = 1, 2, \dots, p$}    
    { 
    Calculate $V^*A^*_{s(j)}$ by DFT of $A^*_{s(j)}$.\;
    Multiply $V^* A_{s(j)}^{*^{(i)}}$ with the eigenvalue $\lambda_i$ for all $i = 1, \dots, 2n$ to construct $\Lambda V^* A_{s(j)}^*$.\;
    Calculate $C(w^*)A^*_{s(j)} = V \Lambda V^* A_{s(j)}^*$ by inverse FFT of $\Lambda V^* A_{s(j)}^*$.\;
    }
 Select the first $n$ rows of $C(w^*)A^*_s$ to form $T(w)A_s$.\;
 Premultiply by $A_s^T$ and divide by $n$.\;
 \KwOut{$\hat{\Sigma}_s$}
 \caption{\cite{heberle2017fast} Algorithm}
 \label{algo:herberle}
\end{algorithm}

We observe that a similar decomposition is possible for the G-SV estimator. Setting $B_s = (Y_{s1} - \bar{\bar{Y}} \; \dots \; Y_{sn} - \bar{\bar{Y}})^T$ and calling Algorithm~\ref{algo:herberle} with inputs $C(w^*)$ and $B_s$, yields $\hat{\Sigma}_{G,s}$. Algorithm~\ref{algo:herberle} has complexity $\mathcal{O}(n \log n p)$ and is thus orders of magnitude faster. Critically, the bandwidth $b_n$ has close to no impact on the computation time.

\section{Effective sample size} \label{sec:ess}

A useful method of assessing the reliability of $\bar{\bar{Y}}$ in estimating $\mu_g$ is effective sample size (ESS). ESS are number of independent and identically distributed samples that would yield the same Monte Carlo variability in $\bar{\bar{Y}}$ as this correlated sample. Let $|\cdot|$ denote determinant. A multiple chain version of the ESS as defined by \cite{vats:fleg:jon:2019} is
\[
\textrm{ESS} = mn \left(\dfrac{|\Upsilon(0)|}{|\Sigma|}\right)^{1/p}\, .
\]
ESS helps users evaluate the quality of their estimation  as it compares to vanilla Monte Carlo, and thus provides easy intuition. In our setting of $m$ parallel chains of $n$ samples each, we estimate ESS with
\[
\widehat{\textrm{ESS}}_G = mn\left(\dfrac{|\hat{\Upsilon}(0)|}{|\hat{\Sigma}_{G}|}\right)^{1/p}\, .
\]
We use the locally-centered $\hat{\Upsilon}(0)$ to estimate $\Upsilon(0)$ instead of $\hat{\Upsilon}_G(0)$ when calculating ESS. This choice controls the overestimation of ESS for slow-mixing Markov chains. Both $\hat{\Upsilon}(0)$ and $\hat{\Upsilon}_G(0)$ are consistent for $\Upsilon(0)$.  For comparison, $\widehat{\textrm{ESS}}_A$ is constructed similarly using $\hat{\Sigma}_A$ instead of $\hat{\Sigma}_G$ to estimate $\Sigma$.  

ESS is employed in determining when to stop an MCMC simulation. \cite{gong:fleg:2016} and \cite{vats:fleg:jon:2019} show that stopping the simulation when the estimated ESS is greater than a pre-specified lower-bound $W_{p}$ yields theoretically valid inference. Further, \cite{vats:knud:2018} establish a one-to-one relationship between ESS and the \cite{gelm:rubi:1992a} potential scale reduction factor, $\hat{R}$. They show that
\[
    \hat{R} = \sqrt{1 + \dfrac{m}{\widehat{\text{ESS}}}}\,.
\]
Different estimates of ESS yield different estimates of the population potential scale reduction factor. Terminating the simulation when  ESS $ > W_p$ is then equivalent to terminating when $\hat{R} \approx 1$. As we will demonstrate in our example, for slow-mixing Markov chains, our proposed estimator $\widehat{\text{ESS}}_G$ provides a considerable improvement over $\widehat{\text{ESS}}_A$, thereby yielding a more robust convergence diagnostic. Here, we use the term ``convergence diagnostic'' to signify convergence of the Monte Carlo estimator, and not the convergence of the Markov chain. All ACF plots in this manuscript have been constructed using the \texttt{R} package \texttt{multichainACF} \footnote{\texttt{https://github.com/medhaaga/multichainACF}} and  reproducible code for the examples is also available publicly\footnote{\texttt{https://github.com/medhaaga/Replicated-Spectral-Variance-Estimator}}.

\section{Examples} \label{sec:examples}

For three different target distributions we sample $m$ parallel Markov chains to assess the performance of our proposed estimators. We make the following three comparisons - (1) locally-centered ACF vs G-ACF estimator, (2) A-SV vs G-SV estimator, and (3) $\widehat{\textrm{ESS}}_A$ vs $\widehat{\textrm{ESS}}_G$. The quality of estimation of $\Sigma$ is studied by coverage probabilities of a $95 \%$ Wald  confidence region when the true mean $\mu_g$ is known. The convergence of local and global estimators of $\Sigma$ and $\textrm{ESS}$ as $n$ increases is studied through two types of running plots (1) logarithm of Frobenius norm of estimated $\Sigma$, and (2) logarithm of estimated ESS/$mn$. In all three examples, we estimate the mean of the stationary distribution, so $g$ is the identity function. The truncation point $b_n$ are the defaults available in the \texttt{R} package \texttt{mcmcse} \citep{mcmcse}. A variety of other examples and simulations are available in the supplement.

\subsection{Vector autoregressive process} \label{ex:var}

Our first example is one where the $\Upsilon$ and $\Sigma$ are available in closed-form, to truly assess the quality of estimation. Consider a $p$-dimensional VAR(1) process $\{X_t\}_{t \geq 1}$ such that
\begin{equation*}
X_t = \Xi X_{t-1} + \epsilon_t\,,  
\end{equation*}
where $X_t \in \mathbb{R}^p$, $\Xi $ is a $p \times p $ matrix, $ \epsilon_t \overset{\text{iid}}{\sim} N(0, \Omega)$, and $\Omega$ is a positive-definite $p \times p$ matrix. We fix $\Omega$ to be an AR correlation matrix with parameter .9. The invariant distribution for this Markov chain is $N(0, \Psi)$, where $vec(\Psi) = (I_{p^2} - \Xi \times \Xi)^{-1} vec(\Omega)$.  For $k \geq 0$, the lag-$k$ autocovariance is $\Upsilon(k) = \Gamma(k) = \Xi^k\Xi$. The process satisfies a CLT if the spectral norm of $\Xi$ is less than 1 \citep{tjos:1990} and the limiting covariance, $\Sigma$, is known in closed form \citep{dai:jon:2017}. We set $p = 2$ and set $\Xi$ to have eigenvalues .999 and .001. 
We further set $m = 5$ with starting values dispersed across the state space.  

We compare the the locally and globally-centered autocorrelations against the truth. In Figure~\ref{fig:var-acf} are the estimated ACF plots for the first component of the second chain against the truth in red. For a run length of $10^3$ (top row), the commonly used locally-centered ACF underestimates the true correlation giving a false sense of security about the mixing of the chain. The G-ACF estimator, on the other hand, is far more accurate. This difference is negligible at the larger run length of $n = 10^4$ when each of the 5 chains have sufficiently explored the state space.

\begin{figure}[h]
\centering
   \includegraphics[width=.60\linewidth]{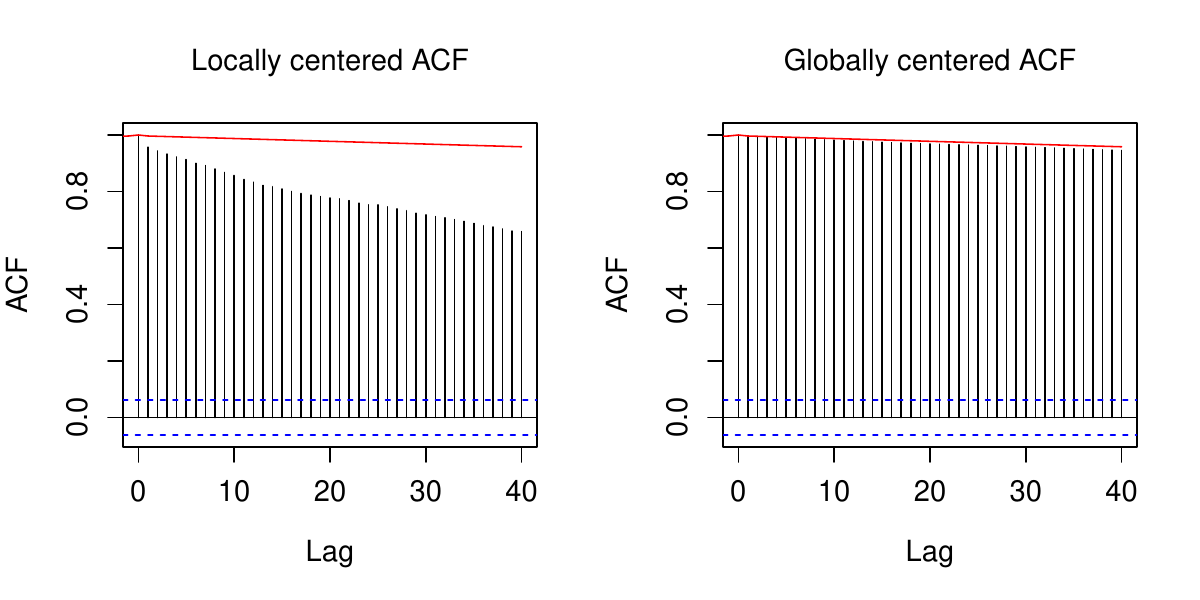}
     \includegraphics[width=.60\linewidth]{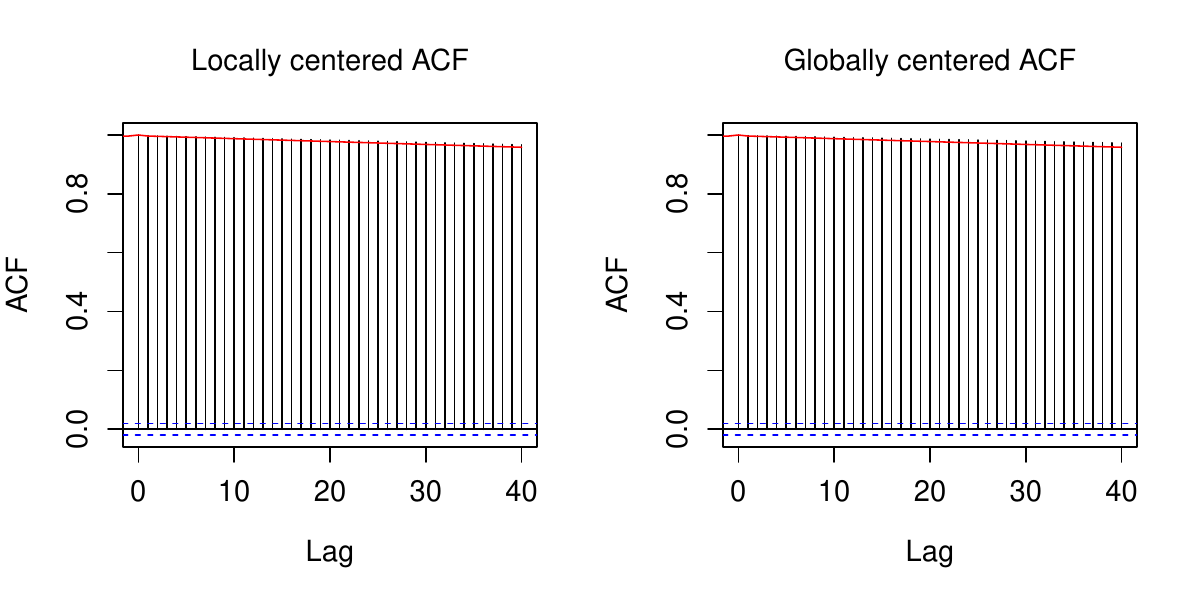} 
\caption{VAR. ACF (left) and G-ACF (right) for the second chain for $m = 5$. (Top) $n = 10^3$  and (bottom) $n =  10^4$. The red line is the true ACF.}
\label{fig:var-acf}
\end{figure}

Since $\mu$ is known, we assess the performance of A-SV and G-SV estimatiors by assessing coverage probabilities of $95\%$ Wald confidence regions over 1000 replications.  Table~\ref{tab:var-coverage} shows that irrespective of sample size, $\hat{\Sigma}_G$ results in close to nominal coverage probability, whereas $\hat{\Sigma}_A$ yields critically low coverage. The low coverage is a consequence of underestimating the autocovariances. Only at sample size $n=10^5$ does $\hat{\Sigma}_A$ yield close to nominal coverage.

\begin{table}[h]
    \centering
    \small
    \begin{tabular}{|c|ccccc|}
    \hline
$n$  &  1000  & 5000  & 10000  & 50000  & 100000 \\  \hline
A-SV  &    0.710 & 0.843 & .885 & .928 & .944 \\ 
G-SV  &    0.956 & 0.937 & .924 & .945 & .952 \\ \hline
    \end{tabular}
    \caption{VAR. Coverage probabilities at $95 \%$ nominal level. Replications $= 1000$.}
    \label{tab:var-coverage}
\end{table}

The quality of estimation of $\Sigma$ and $\text{ESS}$ is assessed by running plots from 50 replications of run length 50000. 
 In Figure~\ref{fig:var-frob_n_ess}, we present running plots of $\log(\|\hat{\Sigma}\|_F)$ for both $\hat{\Sigma}_G$ and $\hat{\Sigma}_A$ and the running plots of $\log(\widehat{\text{ESS}})/mn$ for both  $\widehat{\text{ESS}}_G$ and $\widehat{\text{ESS}}_A$. It is evident that $\hat{\Sigma}_A$ severely underestimates the truth, leading to an overestimation of ESS. G-SV estimates $\Sigma$ more accurately early on, safeguarding against early termination using ESS. 

\begin{figure}[h]
    \centering
      \includegraphics[width = 2.4in]{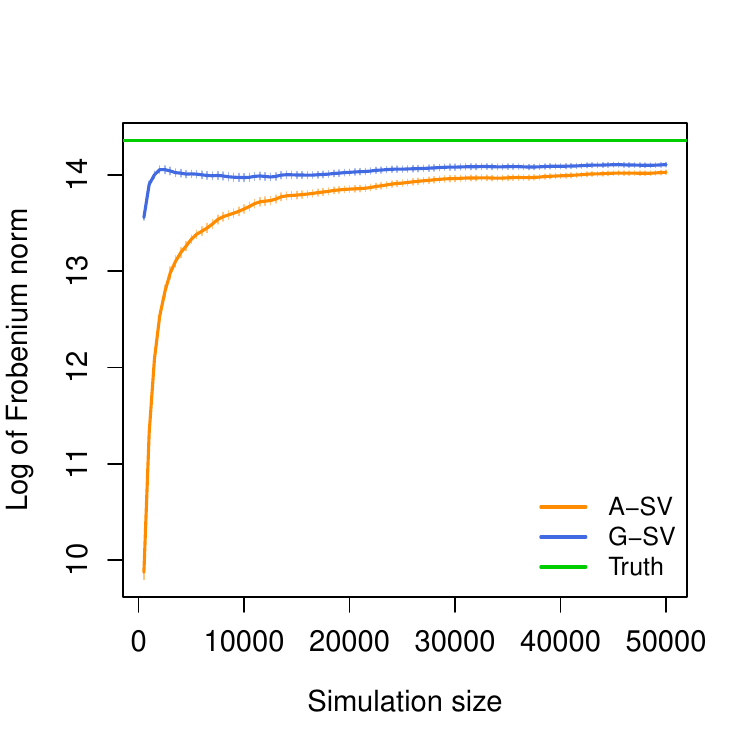} 
      \includegraphics[width = 2.4in]{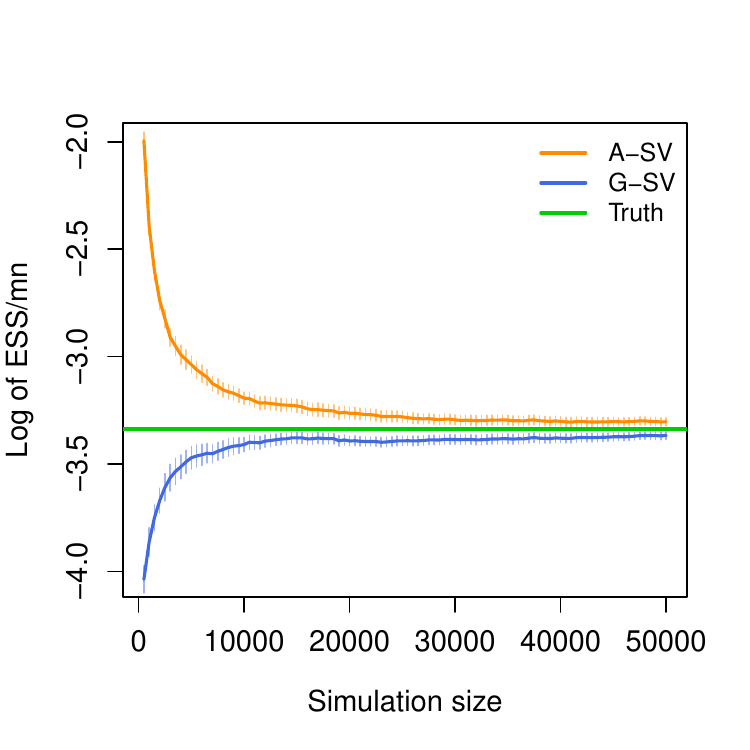}         
    \caption{VAR. (Left) Running plot for logarithm of Frobenius norm of A-SV and G-SV estimator. (Right) Running plot for logarithm of $\widehat{\textrm{ESS}}/mn$ using A-SV and G-SV estimator.}
    \label{fig:var-frob_n_ess}
\end{figure}

\subsection{Boomerang target distribution} \label{ex:boomerang}

Consider the following family of bimodal bivariate distributions introduced by \cite{gelman1991note}, which we term as a boomerang distribution. For $A \geq 0$ and $B, C  \in \mathbb{R}$, the target density is
\[
f(x, y) \propto \exp\left(-\dfrac{1}{2} \left[Ax^2y^2 + x^2 + y^2 -2Bxy  -2Cx - 2Cy  \right]\right)\,.
\]
We run a deterministic scan Gibbs sampler using the following full conditional densities:
\begin{align*}
    x \mid y &\sim N\left(\dfrac{By + C}{Ay^2 + 1}, \dfrac{1}{Ay^2 + 1}\right)\\
    y \mid x &\sim N\left(\dfrac{Bx + C}{Ax^2 + 1}, \dfrac{1}{Ax^2 + 1}\right)\,.
\end{align*}

We consider two settings; in setting 1,
$A = 1,\; B = 3,\; C = 8$ which results in well-separated modes. In setting 2, we let $A = 1, \; B = 10, \; C=7$ which yields a boomerang shape for the contours. The contour plots for these two settings are in the left in Figure~\ref{fig:boom-2D}, overlaid with scatter plots of two parallel runs of the Gibbs sampler. Setting 2 is chosen specifically to illustrate that the locally and globally-centered ACvFs perform similarly when the Markov chain moves freely in the state space. 

\begin{figure}[h]
    \centering
    \includegraphics[width = .90\textwidth]{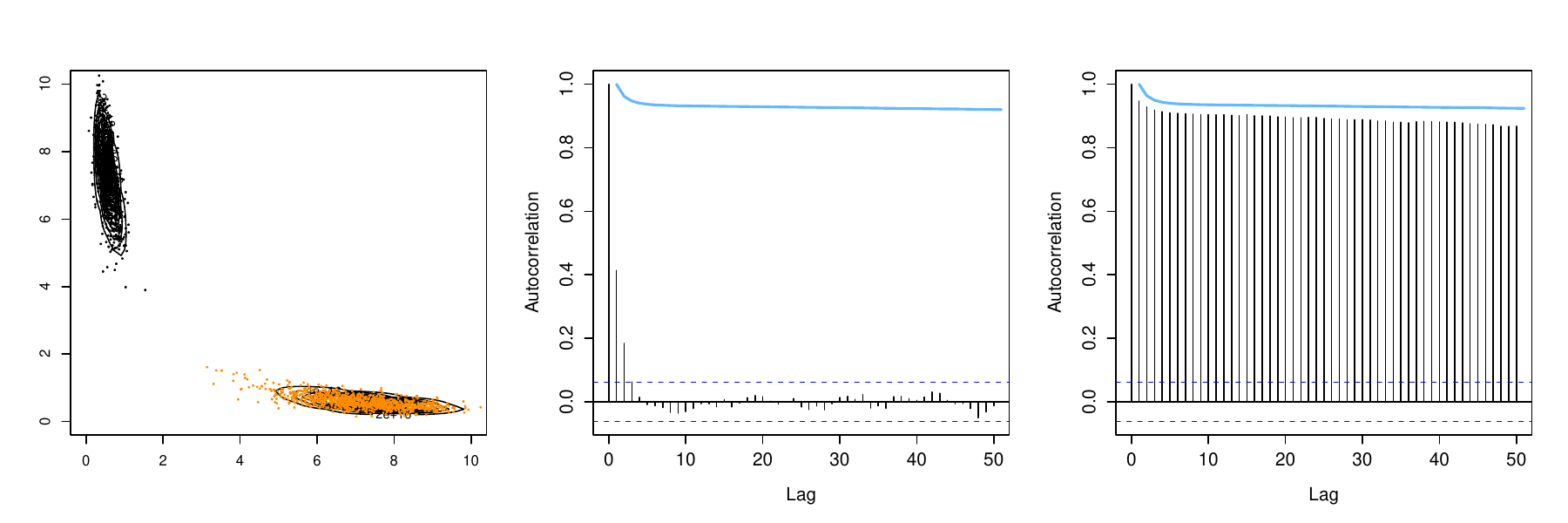}\\
    \includegraphics[width = .90\textwidth]{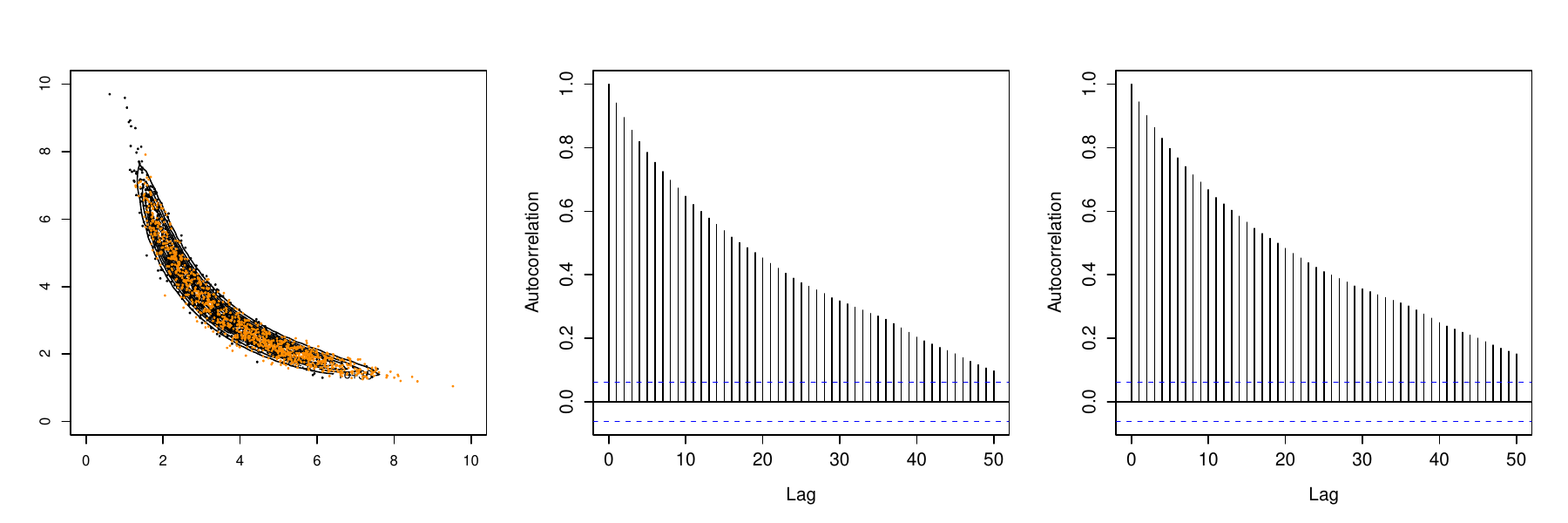}    
    \caption{Boomerang. Top -  Setting 1; bottom -  Setting 2. Contour plots of the target distributions  overlaid with scatter plots for two chains.  Locally-centered ACF (middle) and G-ACF (right) plots for one chain for $n = 10^3$ with the blue line being the ACF at $n = 10^4$.  }
   \label{fig:boom-2D}
\end{figure}

We run $m=5$ parallel Markov chains with starting points evenly distributed across the state space. For setting 1, when the Markov chains have not been able to jump modes, locally-centered ACF severely underestimates autocorrelation. This is seen in the top-middle plot of Figure \ref{fig:boom-2D}, where the locally-centered autocorrelations at $n=1000$ are drastically different from the locally-centered autocorrelations at $n = 10^4$. Somewhere between $n = 1000$ and $n = 10^4$, the Markov chains jump modes and it is only then that the locally-centered ACFs provide better estimates. The G-ACFs,  on the other hand, produce similar ACF estimates at $n = 1000$ and $n = 10^4$ by measuring deviations about the global mean. For setting 2, at $n = 1000$, both methods yield similar ACFs  reinforcing our claim that there is much to gain by using G-ACvF and nothing to lose.

\begin{table}[h]
\parbox{.45\linewidth}{
\centering
\small
\begin{tabular}{|c|c|c|c|c|}
\hline
 $n$ & \multicolumn{2}{|c|}{$m = 2$} & \multicolumn{2}{|c|}{$m=5$}\\
 \hline
 & A-SV & G-SV & A-SV & G-SV \\
 \hline
 5000 & 0.595 &  0.700 & 0.402 &  0.640\\
 10000 & 0.563 &  0.665 & 0.59 &  0.739\\
 50000 & 0.775 &  0.814 & 0.807 &  0.864\\
 100000 & 0.847 &  0.864 & 0.884 &  0.902\\
\hline
\end{tabular}
\caption{Setting 1: Coverage probabilities from $10^3$ replications.}
\label{tab:boom-coverage_1}
}
\hfill
\parbox{.45\linewidth}{
\centering
\small
\begin{tabular}{|c|c|c|c|c|}
 \hline
 $n$ & \multicolumn{2}{|c|}{$m = 2$} & \multicolumn{2}{|c|}{$m=5$}\\
 \hline
 & A-SV & G-SV & A-SV & G-SV \\
 \hline
 1000 &  0.856 &  0.868 & 0.895 &  0.910\\
 5000 & 0.921 &  0.925 & 0.910 &  0.915\\
 10000 & 0.928 &  0.93 & 0.919 &  0.926\\
 50000 & 0.943 &  0.944 & 0.951 &  0.952\\
\hline
\end{tabular}
\caption{Setting 2: Coverage probabilities from $10^3$ replications.}
\label{tab:boom-coverage_2}
}
\end{table}

The true mean of the target distribution can be obtained using numerical approximation. Using  the A-SV and G-SV estimators, we construct $95\%$ confidence regions and report coverage probabilities for 1000 replications for both $m=2$ and $m=5$. Tables~\ref{tab:boom-coverage_1} and~\ref{tab:boom-coverage_2} report all  results. In setting 1, systematically, the G-SV estimator yields far superior coverage than the A-SV estimator for all values of $n$. Whereas for setting 2, the results are almost similar indicating the equivalence of A-SV and G-SV estimator for fast mixing Markov chains.

In Figure~\ref{fig:boom-ess} we present running plots of estimates of $\log(\textrm{ESS})/mn$ for both setting 1 and setting 2. As the sample size increases, $\widehat{\text{ESS}}_G$ and  $\widehat{\text{ESS}}_A$ become closer, but early on for setting 1, $\widehat{\text{ESS}}_G$ estimates are much smaller, safeguarding users against early termination.

\begin{figure}[h]
    \centering
      \includegraphics[width = .40\textwidth]{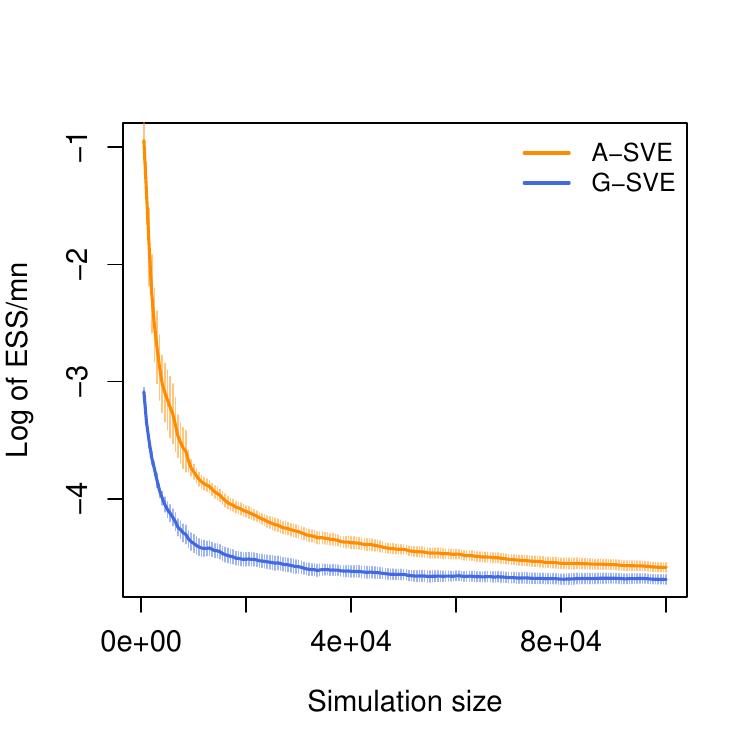}
      \includegraphics[width = .40\textwidth]{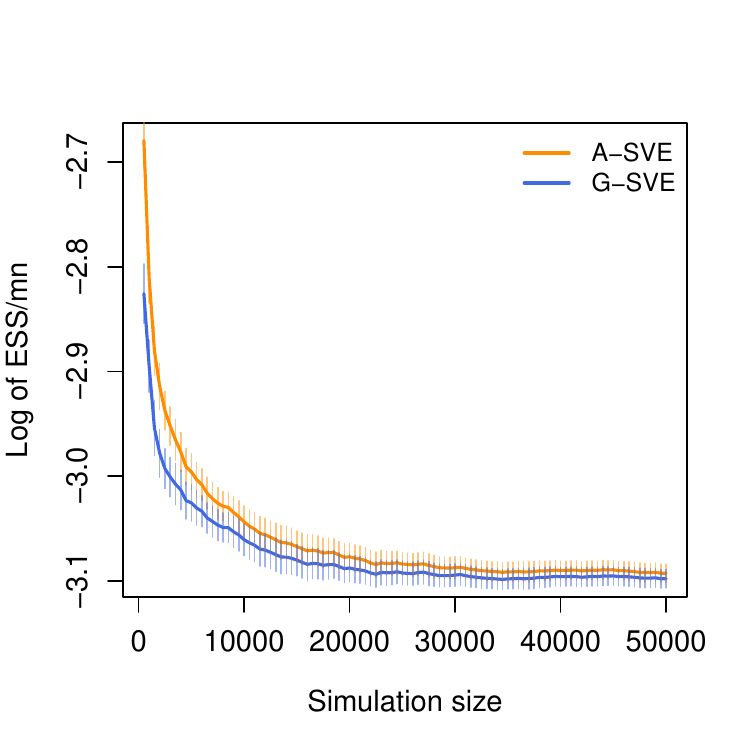}          

    \caption{Boomerang. Running plot of $\log (\widehat{\text{ESS}}_A )/mn$ and $\log (\widehat{\text{ESS}}_G )/mn$ with $m = 5$ for setting 1 (left) and setting 2 (right).}
    \label{fig:boom-ess}
\end{figure}

\subsection{Sensor network localization}

Consider the sensor network localization problem of \cite{ihler2005nonparametric} where the goal is to identify unknown sensor locations using noisy distance data. We use the data and the model from \cite{tak2018repelling}. There are four sensors scattered on a planar region where $x_i = (x_{i1}, x_{i2})^T$ denotes the coordinates of the $i${th} sensor. Let $y_{ij}$ denote the distance between  sensors $x_i$ and $x_j$ and if observed, yields $z_{ij} = 1$ otherwise, $z_{ij} = 0$. The complete model is,
\begin{align*}
    z_{ij} \mid x_1, ..., x_4 & \sim \text{Bernoulli}\left(\exp\left(\dfrac{-\|x_i - x_j\|^2}{2R^2}\right)\right)\\
    y_{ij} \mid z_{ij} = 1, x_i,x_j &\sim N \left(\|x_i - x_j\|^2, \sigma^2 \right)\,.
\end{align*}
\cite{tak2018repelling} set $R = 0.3$ and $\sigma = 0.02$ and use independent $N(0, 100I_2)$ priors on the locations. Distance $y_{ij}$ is specified only if $z_{ij} = 1$. The $8$-dimensional posterior of $(x_1, x_2, x_3, x_4)$ is intractable with unknown full conditionals. A Metropolis-within-Gibbs type sampler is implemented with each full conditional employing the repelling attractive Metropolis (RAM) algorithm of \cite{tak2018repelling}. The RAM algorithm runs Markov chains with higher jumping frequency between the modes.

We run $m = 5$ parallel Markov chains with well-separated starting points. Coverage probabilities are not estimable since the true posterior mean is unknown. Trace plot of $x_{11}$ for two Markov chains is shown in Figure~\ref{fig:sensor-trace}. Figure~\ref{fig:sensor-trace} also  plots locally and globally-centered ACFs. As seen in previous examples, early estimates of locally-centered ACFs are much smaller than later estimates. The G-ACFs, on the other hand, are similar at small and large sample sizes. Figure~\ref{fig:sensor-frob_n_ess} presents the running plot of $\log \|{\Sigma}\|_F$ and $\log \textrm{ESS}/mn$ using A-SV and G-SV  estimators. In both the plots,the  G-SV estimator and $\widehat{\textrm{ESS}}_G/mn$ reach stability significantly earlier than A-SV estimator and $\widehat{\textrm{ESS}}_A/mn$.

\begin{figure}[h]
    \centering
    \includegraphics[width = .32\textwidth]{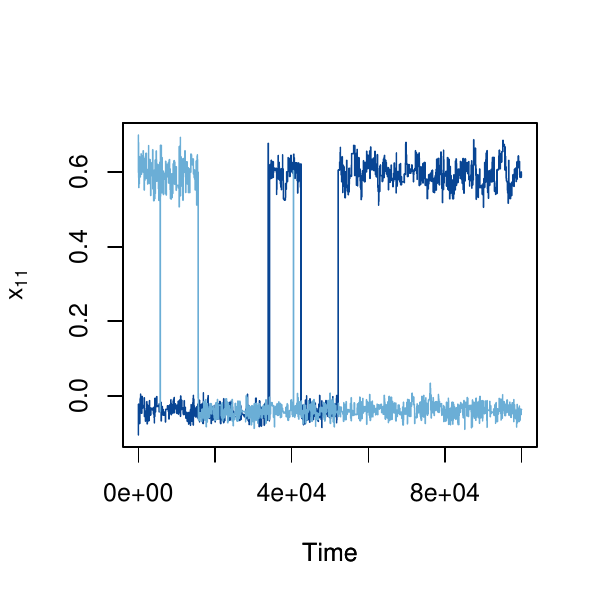}
      \includegraphics[width = .32\textwidth]{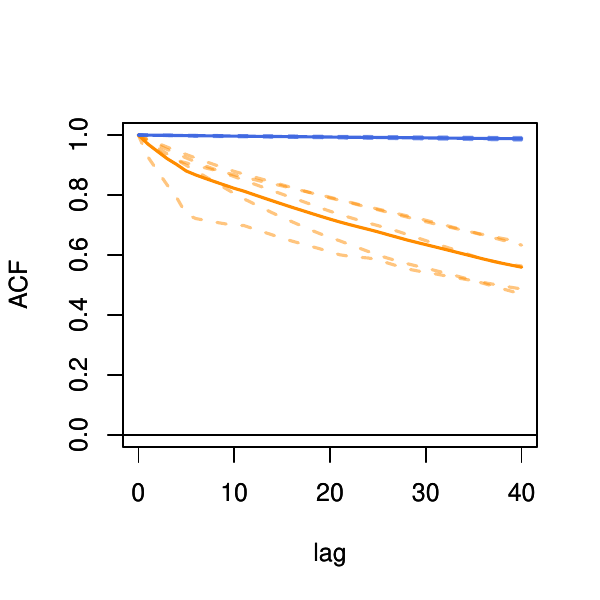}
      \includegraphics[width = .32\textwidth]{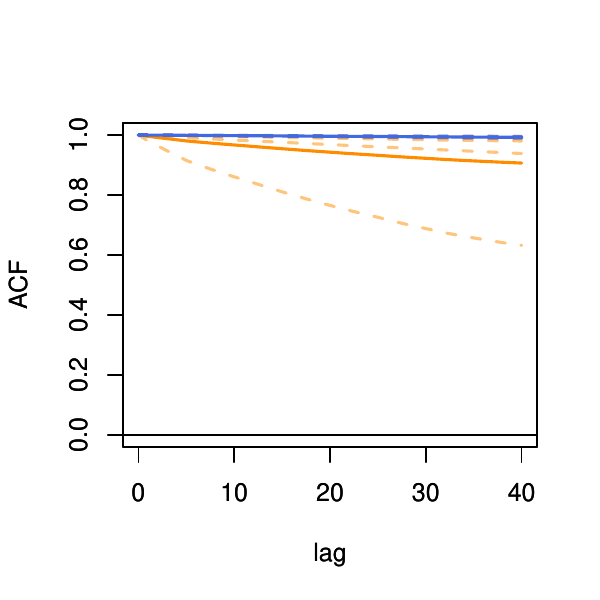}    
    \caption{Sensor: Trace plot of $x_{11}$ (left) for two parallel chains. Average locally-centered ACF (solid orange) and G-ACF (solid blue) at $n = 5000$ (middle) and $n = 50000$ (right). Dashed lines are individual  chain estimates.}
    \label{fig:sensor-trace}
\end{figure}

\begin{figure}[h]
    \centering
      \includegraphics[width = .33\textwidth]{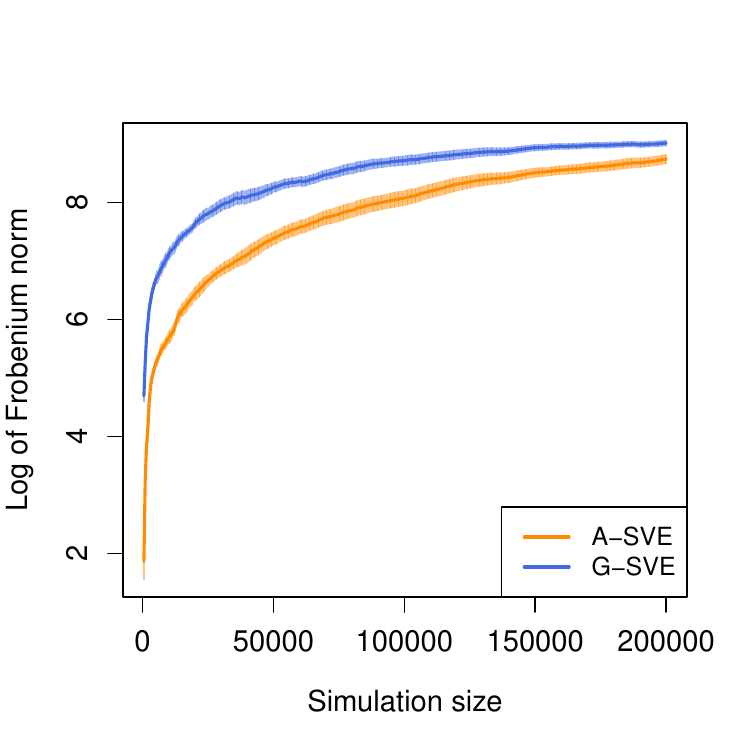} 
      \includegraphics[width = .33\textwidth]{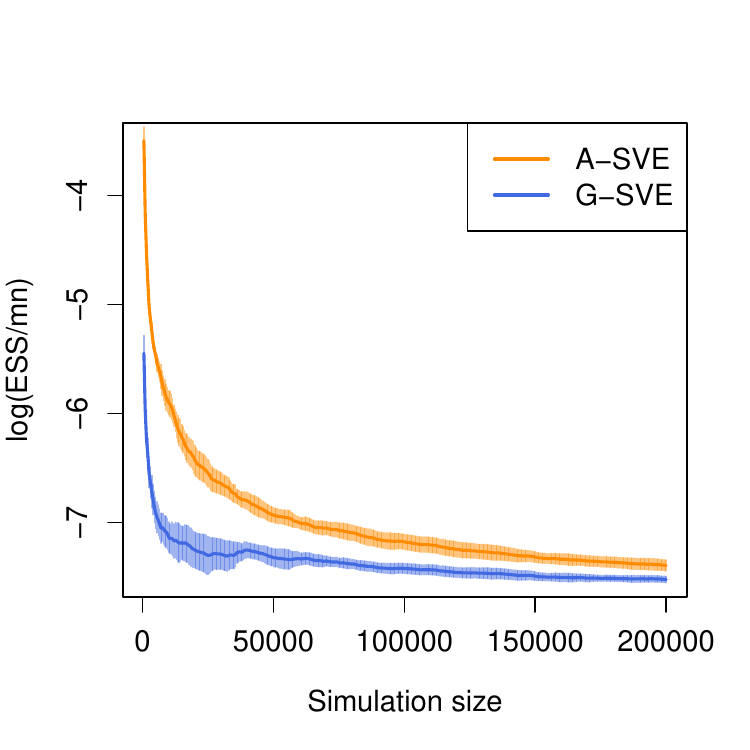}

    \caption{Sensor: Running plot of $\log (\|{\Sigma}\|_F)$ (left) and $\log ({\textrm{ESS}})/mn$ (right) estimated using A-SV and G-SV along with standard errors for 10 replications.}
    \label{fig:sensor-frob_n_ess}
\end{figure}
 
\section{Discussion} \label{sec:discussion}

For slow-mixing Markov chains, a naive average of locally-centered ACvFs can dramatically underestimate the truth. This has a severe impact on ACF plots, Monte Carlo variance, and stopping time of MCMC algorithms. We provide a globally-centered estimate of the ACvF that leads to improvements in all three aspects of MCMC output analysis. 

Another class of estimators of $\Sigma$, for reversible Markov chains, are the multivariate initial sequence (mIS) estimators \citep{dai:jon:2017,koso:2000}. Similar to SV estimators, ACvFs are a critical part of mIS estimators and underestimation in ACvFs yields underestimates of $\Sigma$. It is easy to show that $\hat{\Upsilon}_{G,s}(k)$ is a strongly consistent estimators for $\Upsilon(k)$. Replacing $\hat{\Upsilon}_s(k)$ for $\hat{\Upsilon}_{G,s}(k)$ and following \cite[Theorem~2]{dai:jon:2017} will yield consistent overestimation of the generalized variance , $\text{det}({\Sigma})$ resulting in a globally-centered version of the mIS estimator.  A valuable line of research would be to study the theoretical and empirical properties of the globally-centered mIS estimators. One disadvantage of using G-ACvFs in mIS and in general, is that in a parallel and automated setup, G-ACvFs would require more communication between cores to obtain the global mean. This may increase computation time considerably. For simulations parallelized via popular softwares like JAGS, WinBUGS, Rstan, automatic checks for simulation termination are typically not built-in. Here, using G-ACvFs adds minimal computational burden. 

All examples presented here are in low to medium dimensions. Estimating $\Sigma$ for large $p$ is a challenging problem. In the supplement, we add a simulation study for estimating $\Sigma$ in a high-dimensional VAR model. We find that, although G-SV consistently performs better than A-SV, both require large Monte Carlo sample sizes to estimate $\Sigma$ well. This is particularly true for slow-mixing Markov chains. Estimating $\Sigma$ accurately for high-dimensional slow-mixing MCMC algorithms is a critical area of future work.

\section{Acknowledgements} 
\label{sec:acknowledgements}
The authors are thankful to Haema Nilakanta for sharing useful \texttt{R} code.

\appendix
\section{Theoretical results}  \label{sec:appendix}
\subsection{Preliminaries} \label{apdx:preliminaries}

Let $B^{(i)}$ denote the $i$th component of the $p$-dimensional standard Brownian motion $B$.
\begin{lemma}
\label{lemma: brownian}
(\cite{csorgo1981strong}). Let $b_n$ be a sequence such that $b_n$ and $n/b_n \to \infty$ as $n \to \infty$. Then for all $\epsilon > 0$ and for almost all sample paths, there exists $n_{0}\left(\epsilon\right)$ such that $\forall n\geq n_{0}(\epsilon)$
\begin{align*}
 \sup_{0\leq t \leq n-b_n}\sup_{0 \leq s \leq b_n} \left| B^{\left(i\right)}\left(t+s\right) - B^{\left(i\right)}\left(t\right) \right| & < \left(1+ \epsilon\right)\left(2b_n\left(\log\dfrac{n}{b_n} + \log\; \log\; n\right)\right)^{1/2} , \\ 
  \sup_{0 \leq s \leq b_n} \left|B^{\left(i\right)}\left(n\right) - B^{\left(i\right)}\left(n - s\right)\right|& < \left(1+ \epsilon\right)\left(2b_n\left(\log\dfrac{n}{b_n} + \log\;\log\;n\right)\right)^{1/2} , \; \textrm{and} \\ 
  \left|B^{\left(i\right)}\left(n\right)\right|& < \left(1+\epsilon\right)\sqrt{2n\;\log \log n} \; . 
\end{align*}
\end{lemma}


In order to generalize summations for positive and negative lags in the following proofs, we define the sets $I_k, J_{k1} \text{ and } J_{k2}$ as $ I_k := \{1, \dots, n-k\}, J_{k1} := \{n-k+1, \dots , n\}, J_{k2} := \{1, \dots, k\}$ for $k > 0$ and  $I_k := \{1-k, \dots, n\}, J_{k1} := \{1, \dots , -k\}, J_{k2} := \{n+k+1, \dots, n\}$ for $k < 0$. For $k=0$, $J_{k1}$ and $J_{k2}$ are empty sets and $I_k := \{1, \dots, n\}$. Notice that the empirical autocovariance estimator at lag-$k$ requires a summation over $I_k$. 

\subsection{Proof of Theorem~\ref{th:G-ACF_bias}} 
\label{appendix:bias}


We can break $\hat{\Gamma}_{G,s}$ into four parts for all $k \geq 1$ as:
\begin{align*}
\hat{\Gamma}_{G,s}(k) & = \dfrac{1}{n}\sum_{t \in I_k} \left(X_{s,t} - \bar{\bar{X}} \right) \left(X_{s,t+k} - \bar{\bar{X}} \right)^T \\
    &= \left[\dfrac{1}{n} \sum_{t \in I_k} \left(X_{s,t} - \bar{X}_s \right) \left(X_{s,t+k} - \bar{X}_s \right)^T \right] + \left[\dfrac{1}{n} \sum_{t \in J_{k1}}  \left( \bar{X}_s - X_{s,t} \right)  \left(\bar{X}_s - \bar{\bar{X}} \right)^T\right]\\ 
    & \quad + \left[\dfrac{1}{n} \sum_{t \in J_{k2}} \left(\bar{X}_s - \bar{\bar{X}} \right)  \left(\bar{X}_s - X_{s,t} \right)^T\right]   + \left[\dfrac{n-\abs{k}}{n} \left(\bar{X}_s - \bar{\bar{X}}\right)   \left(\bar{X}_s - \bar{\bar{X}} \right)^T\right]\\
    &= \hat{\Gamma}_s(k) - \dfrac{1}{n}\sum_{t \in J_{k1}}A_{s,t} - \dfrac{1}{n} \sum_{t \in J_{k2}}A_{s,t}^T  + \dfrac{n-\abs{k}}{n}  \left(\bar{X}_s - \bar{\bar{X}} \right)  \left(\bar{X}_s - \bar{\bar{X}} \right)^T\,, \numberthis \label{eq:acf_breakdown}
\end{align*}
where $A_{s,t} = (X_{s,t}-\bar{X}_s)(\bar{X}_s - \bar{\bar{X}})^T$. Under the assumption of stationarity, we will study the expectations of each of the above terms. Without loss of generality, consider $A_{1,1}$,
\begin{align*}
    &\E \left[A_{1,1} \right]\\
     &= \mathbb{E} \left[ \left(X_{1,1} - \bar{X}_1 \right) \left(\bar{X}_1 - \bar{\bar{X}} \right)^T \right]\\
    &= \mathbb{E} \left[X_{1,1}\bar{X}_1^T \right] - \dfrac{1}{m}\mathbb{E} \left[X_{1,1}\bar{X}_1^T \right] - \dfrac{m-1}{m}\mathbb{E} \left[X_{1,1}\bar{X}_2^T \right] + \dfrac{1}{m}\mathbb{E} \left[\bar{X}_1\bar{X}_1^T \right] + \dfrac{m-1}{m}\mathbb{E} \left[\bar{X}_1\bar{X}_2^T \right] - \mathbb{E} \left[\bar{X}_1\bar{X}_1^T \right]\\
    &= \dfrac{m-1}{m}\left(\mathbb{E} \left[X_{1,1}\bar{X}_1^T \right] - \mathbb{E} \left[X_{1,1}\bar{X}_2^T \right] + \mathbb{E} \left[\bar{X}_1\bar{X}_2^T \right] - \mathbb{E} \left[\bar{X}_1\bar{X}_1^T \right]\right)\\
    &= \dfrac{m-1}{m}\left(\dfrac{1}{n}\sum_{t=1}^{n}\mathbb{E} \left[X_{1,1}X_{1,t}^T \right] - \mathbb{E}\left[X_{1,1}\right] \mathbb{E}\left[\bar{X}_2^T \right] + \mathbb{E} \left[\bar{X}_1\right]\mathbb{E}\left[\bar{X}_2^T\right] - \text{Var}\left[\bar{X}_1 \right] - \mathbb{E}\left[\bar{X}_1 \right]\mathbb{E}\left[\bar{X}_1^T \right]\right)\\
    &= \dfrac{m-1}{mn}\left(\sum_{k=0}^{n-1}\Gamma(k) - n\text{Var}\left[\bar{X}_1 \right]\right)\,. \numberthis \label{eq:acf_2}
\end{align*}
Similarly,
\begin{align}
\label{eq:acf_3}
    \mathbb{E} \left[ A_{1,1}^T \right] &= \mathbb{E}\left[ A_{1,1}\right]^T = \dfrac{m-1}{mn}\left(\sum_{k=0}^{n-1}\Gamma(k)^T - n\text{Var}\left[\bar{X}_1 \right] \right)\,.
\end{align}

Further,
\begin{align*}
\mathbb{E}\left[ \left(\bar{X}_{1} - \bar{\bar{X}} \right)  \left(\bar{X}_{1} - \bar{\bar{X}} \right)^{T}\right] &= \mathbb{E}\left[ \bar{X}_{1}\bar{X}_{1}^{T} - \bar{X}_{1}\bar{\bar{X}}^{T} - \bar{\bar{X}} \bar{X}_{1}^{T} + \bar{\bar{X}}\bar{\bar{X}}^{T}\right]\\
&= \left(\Var(\bar{X}_{1}) + \mu\mu^{T} - \Var(\bar{\bar{X}}) - \mu\mu^{T}\right)= \dfrac{m-1}{m}\Var(\bar{X}_1)\,. \numberthis \label{eq:acf_4}
\end{align*}
Additionally, $\hat{\Gamma}_s(k)$ exhibits the following expectation (from \cite{priestley1981spectral})
 \begin{equation} \label{eq:priestly}
     \mathbb{E}[\hat{\Gamma}_s(k)] = \left(1- \dfrac{\abs{k}}{n}\right)\left(\Gamma(k) - \Var{(\bar{X}_s)}
 \right)\,.
 \end{equation}
Using\eqref{eq:acf_2}, \eqref{eq:acf_3}, \eqref{eq:acf_4}, and  \eqref{eq:priestly} in \eqref{eq:acf_breakdown},
\begin{align*}
    & \E \left[\hat{\Gamma}_{G,s}(k) \right] \\
    &= \mathbb{E}\left[\hat{\Gamma}_{s}(k)\right] - \dfrac{1}{n} \left(\sum\limits_{t \in J_{k1}}\mathbb{E}[A_{1,t}] + \sum\limits_{t \in J_{k2}}\mathbb{E}[A_{1,t}^T] \right) + \left(1- \dfrac{\abs{k}}{n}\right)\left(1-\dfrac{1}{m}\right)\Var(\bar{X}_1)\\
    &= \mathbb{E}\left[\hat{\Gamma}_{s}(k)\right] - \dfrac{\abs{k}}{n}\left(1-\dfrac{1}{m}\right)\left(\dfrac{1}{n}\sum_{h=0}^{n-1}\Gamma(h) + \dfrac{1}{n}\sum_{h=0}^{n-1}\Gamma(h)^T - 2 \Var(\bar{X}_1)\right) + \left(1- \dfrac{\abs{k}}{n}\right)\left(1-\dfrac{1}{m}\right)\Var(\bar{X}_1)\\
   &= \left(1- \dfrac{\abs{k}}{n}\right)\left( \Gamma(k) - \dfrac{\Var (\bar{X}_1)}{m} \right) + o(n^{-1})
\end{align*}

By \cite[Proposition 1]{song1995optimal},
\[
\Var(\bar{X}_s) = \dfrac{\Sigma}{n} + \dfrac{\Phi}{n^2} + o(n^{-2})\,.
\]
As a consequence, we get the result 
\[
 \E \left[\hat{\Gamma}_{G,s}(k) \right] =  \left(1- \dfrac{\abs{k}}{n}\right)\left( \Gamma(k) - \dfrac{\Sigma}{mn} - \dfrac{\Phi}{mn^2} \right) + o(n^{-1})\,.
\]

We can also compute the following difference in expectation:
\begin{align*}
  & \E \left[\hat{\Gamma}_{G,s}(k) \right] -  \E \left[\hat{\Gamma}_{s}(k) \right] \\
  &= (1-m^{-1}) \left(1 - \dfrac{\abs{k}}{n}\right)\Var(\bar{X}_1)  - (1 - m^{-1}) \dfrac{\abs{k}}{n}\left(\dfrac{1}{n}\sum_{h=0}^{n-1}\Gamma(h) + \dfrac{1}{n}\sum_{h=0}^{n-1}\Gamma(h)^T -  2\Var(\bar{X}_1)\right) \\
  &= (1-m^{-1})\left[ \Var(\bar{X}_1)  -  \dfrac{\abs{k}}{n}\left(\dfrac{1}{n}\sum_{h=0}^{n-1}\Gamma(h) + \dfrac{1}{n}\sum_{h=0}^{n-1}\Gamma(h)^T -  \Var(\bar{X}_1)\right)\right] \\
  &= (1-m^{-1})\left[\dfrac{\Sigma}{n} + \dfrac{\abs{k}}{n^2}\left(\Sigma - \sum_{h=0}^{n-1}\Gamma(h) -\sum_{h=0}^{n-1}\Gamma(h)^T \right) + \left(1+\dfrac{\abs{k}}{n}\right)\dfrac{\Phi}{n^2}\right] + o(n^{-2})\,.
\end{align*}

Although this completes the proof of Theorem~\ref{th:G-ACF_bias}, we will require the following decomposition
 \begin{equation} \label{eq:G-ACF_expec_breakdown}
     \mathbb{E}\left[\hat{\Gamma}_{G,s}(k)\right] = \left(1- \dfrac{\abs{k}}{n}\right)\Gamma(k) + O_1 + O_2\,.
 \end{equation}
where,
\begin{align*}
    O_1 &= -\dfrac{\abs{k}}{n}\left[\left(1-\dfrac{1}{m}\right)\left(\dfrac{1}{n}\sum_{h=0}^{n-1}\Gamma(h)^T + \dfrac{1}{n}\sum_{h=0}^{n-1}\Gamma(h)\right) - \left(2-\dfrac{1}{m}\right) \left(\dfrac{\Sigma}{n} + \dfrac{\Phi}{n^2}\right)\right] + o(n^{-2})\, ,\\
    O_2 &= -\dfrac{1}{m}\left(\dfrac{\Sigma}{n} + \dfrac{\Phi}{n^2}\right) + o(n^{-2})\,.
\end{align*}

\subsection{Strong consistency of the G-SV estimator} \label{appendix:strong_consis}

Consider pseudo autocovariance and spectral variance estimators for the $s$th chain, denoted by $\tilde{\Upsilon}_s(k)$ and $\tilde{\Sigma}_s$ that use data centered around the unobserved actual mean $\mu_g$:
\begin{align*}
    \tilde{\Upsilon}_s(k) &= \dfrac{1}{n}\sum_{t \in I_k}(Y_{s,t}-\mu_g)(Y_{s,t+k}-\mu_g)^T  \quad \text{and} \quad 
    \tilde{\Sigma}_s = \sum_{k=-b_n+1}^{b_n-1}w\left(\dfrac{k}{b_n}\right)\tilde{\Upsilon}_s(k) \,.
\end{align*}
The average pseudo spectral variance estimator is $\tilde{\Sigma}_A = m^{-1}\sum_{s=1}^{m}\tilde{\Sigma}_s$
%
Further, let
\begin{align*}
  M_1 & = \dfrac{1}{m}\sum\limits_{s=1}^{m}\left\{\sum\limits_{k=-b_n+1}^{b_n-1}w\left(\dfrac{k}{b_n}\right)\sum\limits_{t \in I_k}\dfrac{1}{n}\left[ \left(Y_{s,t}-\mu_g \right)_i   \left(\mu_g-\bar{\bar{Y}} \right)_j +    \left(\mu_g-\bar{\bar{Y}} \right)_i  \left(Y_{s,t+k}-\mu_g \right)_j \right]\right\}\,, \\ 
M_2 &= \left(\mu_g-\bar{\bar{Y}} \right)_i   \left(\mu_g-\bar{\bar{Y}} \right)_j\sum\limits_{k=-b_n+1}^{b_n-1}\left(1-\dfrac{\abs{k}}{n}\right)w\left(\dfrac{k}{b_n}\right)\,.
\end{align*}

\begin{lemma} \label{lemma:G-SVE_breakdown}
Let Assumption~\ref{ass:lag_window} holds, then for the G-SV estimator, $\hat{\Sigma}_{G}^{ij} = \tilde{\Sigma}_A^{ij} + M_1 + M_2$ and 
\[
|M_1 + M_2| \leq D^2 g_1(n) + D g_2(n) + g_3(n)\,,
\]
where for some constant $C$,
\begin{align*}
    g_1(n) &= (4c+C)\dfrac{b_n \psi^2(n)}{n^2} - 4c\dfrac{\psi^2(n)}{n^2} \to 0\\
    g_2(n) &= 2\sqrt{2}\|L\|p^{1/2}(1+\epsilon)\left[(4c+C)\dfrac{b_n\psi(n)\sqrt{n\log \log n}}{n^2} - 4c\dfrac{\psi(n)\sqrt{n\log \log n}}{n^2}\right] \to 0\\
    g_3(n) &= \|L\|^2 p (1+\epsilon)^2\left[(4c+C)\dfrac{b_n \log\log n}{n} - 4c \dfrac{\log \log n}{n}\right] \to 0 \quad \text{as }n \to \infty\,.
\end{align*}
\end{lemma}

\begin{proof}
The proof follows from standard algebraic calculations and is presented here for completeness. Consider,
\begin{align*}
\hat{\Sigma}_{G}^{ij} &= \dfrac{1}{m}\sum_{s=1}^{m} \sum_{k=-b_n+1}^{b_n-1}w\left(\dfrac{k}{b_n}\right)\dfrac{1}{n}\sum_{t \in I_k} \left(Y_{s,t}-\bar{\bar{Y}} \right)_i \left(Y_{s,t+k}-\bar{\bar{Y}} \right)_j\\
&= \dfrac{1}{m}\sum_{s=1}^{m}\sum_{k=-b_n+1}^{b_n-1}w\left(\dfrac{k}{b_n}\right)\dfrac{1}{n}\sum_{t \in I_k}  \left[ \left(Y_{s,t}-\mu_g \right)_i  \left(Y_{s,t+k}-\mu_g \right)_j+  \left(Y_{s,t} - \mu_g \right)_i    \left(\mu_g - \bar{\bar{Y}} \right)_j \right. \\  
& \quad + \left. \left(\mu_g-\bar{\bar{Y}} \right)_i  \left(Y_{s,t+k}-\mu_g \right)_j + \left(\mu_g-\bar{\bar{Y}} \right)_i  \left(\mu_g-\bar{\bar{Y}}  \right)_j  \right]\\
& = \tilde{\Sigma}^{ij}_A + \left[(\mu_g-\bar{\bar{Y}})_i(\mu_g-\bar{\bar{Y}})_j\sum_{k=-b_n+1}^{b_n-1}\left(1-\dfrac{\abs{k}}{n}\right)w\left(\dfrac{k}{b_n}\right)\right] \\ 
& \quad  + \dfrac{1}{m}\sum_{s=1}^{m}\sum_{k=-b_n+1}^{b_n-1}  w\left(\dfrac{k}{b_n}\right)\sum_{t \in I_k}  \left[\dfrac{1}{n} \left(Y_{s,t} - \mu_g \right)_i \left(\mu_g - \bar{\bar{Y}} \right)_j + \dfrac{1}{n} \left(\mu_g-\bar{\bar{Y}} \right)_i  \left(Y_{s,t+k}-\mu_g \right)_j \right] \\ 
& =: \tilde{\Sigma}^{ij}_A + M_2 + M_1\,.
\end{align*}

Consequently
\[
\left|\hat{\Sigma}_{G}^{ij} - \tilde{\Sigma}_{A}^{ij}  \right| = |M_1 + M_2| \leq |M_1| + |M_2|\,.
\]
We first present a result which will be useful later. For any Markov chain $s$, 
\begin{align*}
  \|\bar{Y}_s - \mu_g\|_{\infty} & \leq \|\bar{Y}_s - \mu_g\| = \dfrac{1}{n}\left\|\sum_{t=1}^{n}Y_{s,t} - n\mu_g\right\| \\
  & \leq \dfrac{1}{n}\left\|\sum_{t=1}^{n}Y_{s,t} - n \mu_g - L B(n)\right\| + \dfrac{\left\|L B(n)\right\|}{n}\\
  &< \dfrac{D\psi(n)}{n} + \dfrac{\|L B(n)\|}{n}\\
  &< \dfrac{D\psi(n)}{n} + \dfrac{1}{n}\|L\| \left(\sum\limits_{i=1}^{p}|B^{(i)}(n)|^2\right)^{1/2}\\
  & \leq \dfrac{D\psi(n)}{n} + \dfrac{1}{n}\|L\| p^{1/2}(1+\epsilon)\sqrt{2n \log\log n}\,. \numberthis \label{eq:xbars_bound}
\end{align*}
 Similarly,
 \begin{equation}
\label{eq:xbarbar_bound}
   \| \bar{\bar{Y}} - \mu_g\|_{\infty} \leq \dfrac{D\psi(n)}{n} + \dfrac{1}{n}\|L\| p^{1/2}(1+\epsilon)\sqrt{2n \log\log n}\,.
\end{equation}
%
Now consider,
\begin{align*}
& |M_1| \\ 
    & \leq \dfrac{1}{m}\sum_{s=1}^{m}\left\{\sum_{k=-b_n+1}^{b_n-1}\abs{w\left(\dfrac{k}{b_n}\right)}\left[ \dfrac{1}{n}\left|\sum_{t \in I_k}(Y_{s,t}- \mu_g)_i\right|\left|(\mu_g-\bar{\bar{Y}})_j\right|+ \dfrac{1}{n}\left|(\mu_g-\bar{\bar{Y}})_i\right|\left|\sum_{t \in I_k }(Y_{s,t+k}-\mu_g)_j\right|\right]\right\}\\
    & \leq \dfrac{c\|(\bar{\bar{Y}} - \mu_g)\|_{\infty}}{m} \sum_{s=1}^{m}\sum\limits_{k=-b_n+1}^{b_n-1}\left[ \dfrac{1}{n}\left\|\sum_{t \in I_k}(Y_{s,t}-\mu_g)\right\|_{\infty} + \dfrac{1}{n}\left\|\sum_{t \in I_k}(Y_{s,t+k}-\mu_g)\right\|_{\infty} \right]\\
    &\leq \dfrac{c\|(\bar{\bar{Y}} - \mu_g)\|_{\infty}}{m} \\
    & \quad \times \sum_{s=1}^{m}\sum_{k=-b_n+1}^{b_n-1}\left[ \dfrac{1}{n}\left\|\sum_{t \in J_{k1}}(Y_{s,t} - \mu_g) - n(\bar{Y}_s - \mu_g) \right\|_{\infty} + \dfrac{1}{n}\left\|\sum_{t \in J_{k2}}(Y_{s,t} - \mu_g) - n(\bar{Y}_s - \mu_g)\right\|_{\infty} \right]\\
    &\leq \dfrac{c\|(\bar{\bar{Y}} - \mu_g)\|_{\infty}}{m} \sum_{s=1}^{m}\sum\limits_{k=-b_n+1}^{b_n-1}\left[ \dfrac{1}{n}\left\|\sum_{t \in J_{k1}}(Y_{s,t} - \mu_g)\right\|_{\infty} + \dfrac{1}{n}\left\|\sum_{t \in J_{k2}}(Y_{s,t} - \mu_g)\right\|_{\infty} + 2\|\bar{Y}_s - \mu_g\|_{\infty} \right]\\
    & \leq \dfrac{c\|(\bar{\bar{Y}} - \mu_g)\|_{\infty}}{m} \sum\limits_{s=1}^{m}\sum_{k=-b_n + 1}^{b_n-1}   \dfrac{1}{n}\left[\left\|\sum_{t \in J_{k1}}(Y_{s,t} - \mu_g)\right\|_{\infty} + \left\|\sum_{t \in J_{k2}}(Y_{s,t} - \mu_g)\right\|_{\infty} \right]\\
    & \; \;+ 2c(2b_n - 1)\|\bar{\bar{Y}} - \mu_g\|_{\infty}\|\bar{Y}_h - \mu_g\|_{\infty} \textrm{ for some } h \in \{1, \dots, m\} \,.
\end{align*}
Using SIP on summation of $k$ terms, we obtain the following upper bound for $|M_1|$
\begin{align*}
|M_1|
    & < 2c\|(\bar{\bar{Y}} - \mu_g)\|_{\infty} \left[\sum\limits_{k=-b_n + 1}^{b_n-1}\left[ \dfrac{D \psi(k)}{n} + \dfrac{\|L\| p^{1/2}(1+\epsilon)\sqrt{2k \log\log k}}{n}  \right]  + (2b_n - 1) \|\bar{Y}_h - \mu_g\|_{\infty}  \right]\\
    &\leq 2c(2b_n - 1)\|(\bar{\bar{Y}} - \mu_g)\|_{\infty} \left[ \dfrac{D \psi(n)}{n} + \dfrac{\|L\| p^{1/2}(1+\epsilon)\sqrt{n \log\log n}}{n} + \|\bar{Y}_h - \mu_g\|_{\infty}  \right]\\
    & \leq 4c(2b_n - 1)\left[ \dfrac{D \psi(n)}{n} + \dfrac{\|L\| p^{1/2}(1+\epsilon)\sqrt{n \log\log n}}{n}  \right]^2  \quad \text{(by \eqref{eq:xbars_bound} and \eqref{eq:xbarbar_bound})}\,. \numberthis \label{eq:strong_consis_term-1}
\end{align*}
%
 %
%
For $M_2$,
\begin{align*}
   |M_2| & = \left|\dfrac{1}{m}\sum\limits_{s=1}^{m}\left\{ \left(\mu_g - \bar{\bar{Y}} \right)_i  \left(\mu_g - \bar{\bar{Y}} \right)_j\sum_{k=-b_n+1}^{b_n-1}\left(1-\dfrac{\abs{k}}{n}\right)w\left(\dfrac{k}{b_n}\right)\right\}\right|\\
    & \leq\|\bar{\bar{Y}} - \mu_g\|_{\infty}^2\abs{\sum_{k=-b_n+1}^{b_n-1}\left(1-\dfrac{\abs{k}}{n}\right)w\left(\dfrac{k}{b_n}\right)} < \|\bar{\bar{Y}} - \mu_g\|_{\infty}^2\left[\sum_{k=-b_n+1}^{b_n-1}\left|w\left(\dfrac{k}{b_n}\right)\right|\right]\\
    & \leq b_n\|\bar{\bar{Y}} - \mu_g\|_{\infty}^2 \int_{-\infty}^{\infty}|w(x)|dx \\
    & \leq Cb_n\left[ \dfrac{D \psi(n)}{n} + \dfrac{\|L\| p^{1/2}(1+\epsilon)\sqrt{n \log\log n}}{n}  \right]^2 \quad \text{ for some constant $C$ (by \eqref{eq:xbarbar_bound})} \numberthis \label{eq:strong_consis_term-2} \,.
\end{align*}
%
%
%
%
Using \eqref{eq:strong_consis_term-1} and \eqref{eq:strong_consis_term-2},
\begin{align*}
|M_1 + M_2|  & \leq |M_1| + |M_2|
   = D^2g_1(n) + D g_2(n) + g_3(n)\,,
\end{align*}
where
\begin{align*}
    g_1(n) &= (8c + C)\dfrac{b_n \psi^2(n)}{n^2} - 4c\dfrac{\psi^2(n)}{n^2}\\
    g_2(n) &= 2\sqrt{2}\|L\|p^{1/2}(1+\epsilon)\left[(8c + C)\dfrac{b_n\psi(n)\sqrt{n\log \log n}}{n^2} - 4c\dfrac{\psi(n)\sqrt{n\log \log n}}{n^2}\right]\\
    g_3(n) &= \|L\|^2 p (1+\epsilon)^2\left[(8c + C)\dfrac{b_n \log\log n}{n} - 4c \dfrac{\log \log n}{n}\right]\,.
\end{align*}
Under the assumptions of Theorem~\ref{thm:kuelbs}, $\psi(n) = o(n^{1/2})$. Using the law of iterative logarithms (LIL), a tighter bound for $\psi(n)$ is given by \cite{stra:1964} as $o(\sqrt{n\log \log n})$. Since $b_n \log \log n / n \to 0$ as $n \to \infty$, consequently, 
$b_n \psi^2(n)/n^2 \to 0$, $\psi^2(n)/n^2 \to 0$, ${b_n\psi(n)\sqrt{n\log \log n}/n^2} \to 0$, and $\psi(n) \sqrt{n \log \log n}/n^2 \to 0$. Thus, $g_1(n), g_2(n)$ and $g_3(n) \to 0$ as $n \to \infty$.
\end{proof}
%
\begin{proof}[Proof of Theorem \ref{th:consistency}]
We have the following decomposition,
\begin{align*}
\tilde{\Sigma}_{A}^{ij}
    & = \dfrac{1}{m}\sum_{s=1}^{m}\sum_{k=-b_n+1}^{b_n-1}w\left(\dfrac{k}{b_n}\right)\dfrac{1}{n}\sum_{t \in I_k}  \left(Y_{s,t} \pm \bar{Y}_s - \mu_g \right)_i  \left(Y_{s, t + k)} \pm \bar{Y}_s - \mu_g \right)_j\\
    & = \hat{\Sigma}_{A}^{ij} + \dfrac{1}{m}\sum_{s=1}^{m}\sum_{k=-b_n+1}^{b_n-1}w\left(\dfrac{k}{b_n}\right)\dfrac{1}{n}\sum_{t \in I_k}\left[ \left(Y_{s,t} - \bar{Y}_s \right)_i   \left(\bar{Y}_s - \mu_g \right)_j + \left(\bar{Y}_s - \mu_g \right)_i  \left(Y_{s,t+k} - \bar{Y}_s \right)_j\right]\\
    & \quad + \left[\dfrac{1}{m}  \sum_{s=1}^{m}  \left(\bar{Y}_s - \mu_g \right)_i  \left(\bar{Y}_s - \mu_g \right)_j \right]  \left[\sum_{k=-b_n+1}^{b_n-1}w\left(\dfrac{k}{b_n}\right)\left(1 - \dfrac{\abs{k}}{b_n}\right) \right]\\
    & = \hat{\Sigma}_{A}^{ij} + N_1 + N_2 \,,
\end{align*}
where
\begin{align*}
N_1 & = \dfrac{1}{m}\sum_{s=1}^{m}  \sum_{k=-b_n+1}^{b_n-1}  w\left(\dfrac{k}{b_n}\right)\dfrac{1}{n}  \sum_{t \in I_k}  \left[ \left(Y_{s,t} - \bar{Y}_s \right)_i  \left(\bar{Y}_s - \mu_g \right)_j + \left(\bar{Y}_s - \mu_g \right)_i  \left(Y_{s,t+k} - \bar{Y}_s \right)_j\right]\\
N_2 & = \left[\dfrac{1}{m}  \sum_{s=1}^{m}  \left(\bar{Y}_s - \mu_g \right)_i  \left(\bar{Y}_s - \mu_g \right)_j \right]  \left[\sum_{k=-b_n+1}^{b_n-1}w\left(\dfrac{k}{b_n}\right)\left(1 - \dfrac{\abs{k}}{b_n}\right) \right]\;.
\end{align*}
Using the above breakdown and  Lemma~\ref{lemma:G-SVE_breakdown}, 
\begin{align*}
\left|\hat{\Sigma}_{G}^{ij} - \Sigma^{ij} \right| & = \left| \hat{\Sigma}_{A}^{ij} - \Sigma^{ij} + N_1 + N_2 + M_1 + M_2 \right|  \leq \left| \hat{\Sigma}_{A}^{ij} - \Sigma^{ij} \right| +  \left| N_1 \right| +  \left|N_2 \right| + \left| M_1 + M_2 \right| \numberthis \label{eq:G-SVE_full_decomp}
\end{align*}
By the strong consistency of single-chain SV estimator, the first term goes to 0 with probability 1 and by Lemma~\ref{lemma:G-SVE_breakdown}, the third term goes to 0 with probability 1 as $n \to \infty$. It is left to show that $|N_1| \to 0$ and $|N_2| \to 0$ with probability 1. Consider,
\begin{align*}
|N_1| & = \left|\dfrac{1}{m}\sum_{s=1}^{m}  \sum_{k=-b_n+1}^{b_n-1}  w\left(\dfrac{k}{b_n}\right)\dfrac{1}{n}  \sum_{t \in I_k}  \left[ \left(Y_{s,t} - \bar{Y}_s \right)_i  \left(\bar{Y}_s - \mu_g \right)_j + \left(\bar{Y}_s - \mu_g \right)_i  \left(Y_{s,t+k} - \bar{Y}_s \right)_j\right] \right|\\
& \leq \left| \dfrac{1}{m}\sum_{s=1}^{m}\sum_{k=-b_n+1}^{b_n-1}w\left(\dfrac{k}{b_n}\right)\dfrac{1}{n}\sum_{t \in I_k}\left[ \left(Y_{s,t} - \bar{Y}_s \right)_i  \left(\bar{Y}_s - \mu_g \right)_j \right] \right| \\ 
& \quad + \left| \dfrac{1}{m}\sum_{s=1}^{m}\sum_{k=-b_n+1}^{b_n-1}w\left(\dfrac{k}{b_n}\right)\dfrac{1}{n}\sum_{t \in I_k}\left[ \left(\bar{Y}_s - \mu_g \right)_i  \left(Y_{s,t+k} - \bar{Y}_s \right)_j\right] \right|\,.
\end{align*}
We will show that the first term goes to 0 and the proof for the second term is similar. Consider
\begin{align*}
    & \left| \dfrac{1}{m}\sum_{s=1}^{m}\sum_{k=-b_n+1}^{b_n-1}w\left(\dfrac{k}{b_n}\right)\dfrac{1}{n}\sum_{t \in I_k}\left[ \left(Y_{s,t} - \bar{Y}_s \right)_i  \left(\bar{Y}_s - \mu_g \right)_j \right] \right| \\
    & \leq \dfrac{1}{m}\sum_{s=1}^{m}\sum_{k=-b_n+1}^{b_n-1} \left| w\left(\dfrac{k}{b_n}\right) \right| \dfrac{ \left| \left(\bar{Y}_s - \mu_g \right)_j \right| }{n} \left| \sum_{t \in J_{k1}} \left(\mu_g - Y_{s,t} \right)_i  + \abs{k} \left(\bar{Y}_s - \mu_g \right)_i \right|\\
    & \leq \dfrac{1}{m}\sum_{s=1}^{m}\sum_{k=-b_n+1}^{b_n-1}  \left|w\left(\dfrac{k}{b_n}\right) \right|\dfrac{\|\bar{Y}_s - \mu_g\|_{\infty}}{n} \left\|\sum_{t \in J_{k1}} \left(\mu_g - Y_{s,t} \right) + \abs{k} \left(\bar{Y}_s - \mu_g \right)\right\|_{\infty}\\
    & \leq \dfrac{1}{m}\sum_{s=1}^{m}\sum_{k=-b_n+1}^{b_n-1} \left|w\left(\dfrac{k}{b_n}\right) \right|\dfrac{\|\bar{Y}_s - \mu_g\|_{\infty}}{n} \left(\left\|\sum_{t \in J_{k1}}  \left(Y_{s,t} - \mu_g \right)\right\|_{\infty} + \abs{k}\|\bar{Y}_s - \mu_g\|_{\infty} \right)\\
    & \leq \dfrac{c}{m}\sum_{s=1}^{m}\sum_{k=-b_n+1}^{b_n-1}\dfrac{\|\bar{Y}_s - \mu_g\|_{\infty}}{n} \left\|\sum_{t \in J_{k1}} \left(Y_{s,t} - \mu_g\right) \right\| + \dfrac{c}{m}\sum\limits_{s=1}^{m} \dfrac{b_n(b_n-1)}{n} \left \|\bar{Y}_s - \mu_g \right \|_{\infty}^2\,.
\end{align*}

Using SIP on the summation of $k$ terms,
\begin{align*}
    & \left|\dfrac{1}{m}\sum_{s=1}^{m}\sum_{k=-b_n+1}^{b_n-1}w\left(\dfrac{k}{b_n}\right)\dfrac{1}{n}\sum_{t=1}^{n-k}\left[ \left(Y_{s,t} - \bar{Y}_s \right)_i  \left(\bar{Y}_s - \mu_g \right)_j \right] \right|\\
   &  < \dfrac{c}{m}\sum\limits_{s=1}^{m}\|\bar{Y}_s - \mu_g\|_{\infty}\sum\limits_{k=-b_n + 1}^{b_n-1}\left[ \dfrac{D \psi(k)}{n} + \dfrac{\|L\| p^{1/2}(1+\epsilon)\sqrt{2k \log\log k}}{n}  \right] + \dfrac{c}{m}\sum\limits_{s=1}^{m} \dfrac{b_n(b_n-1)}{n}\|\bar{Y}_s - \mu_g\|_{\infty}^2\\
   &  < \dfrac{c(2b_n-1)}{m}\sum\limits_{s=1}^{m}\|\bar{Y}_s - \mu_g\|_{\infty}\left[ \dfrac{D \psi(n)}{n} + \dfrac{\|L\| p^{1/2}(1+\epsilon)\sqrt{2n \log\log n}}{n}  \right] + \dfrac{c}{m}\sum\limits_{s=1}^{m} \dfrac{b_n(b_n-1)}{n}\|\bar{Y}_s - \mu_g\|_{\infty}^2 \\
   &  \leq   c\left(2b_n - 1 + \dfrac{b_n^2}{n} - \dfrac{b_n}{n}\right)\left[ \dfrac{D \psi(n)}{n} + \dfrac{\|L\| p^{1/2}(1+\epsilon)\sqrt{2n \log\log n}}{n}  \right]^2  \to 0\,. \quad  \text{(by \eqref{eq:xbars_bound})}
\end{align*}
%
%
%
%
Similarly, the second part of $N_1 \to 0$ with probability 1. 
%
Following the steps in \eqref{eq:strong_consis_term-2}, 
\begin{align*}
    |N_2| & \leq Cb_n\left[ \dfrac{D \psi(n)}{n} + \dfrac{\|L\| p^{1/2}(1+\epsilon)\sqrt{2n \log\log n}}{n}  \right]^2  \to 0\,.
\end{align*}
Thus, in \eqref{eq:G-SVE_full_decomp}, every term goes to 0 and $\hat{\Sigma}_{G}^{ij} \to \Sigma^{ij}$ with probability 1 as $n \to \infty$. 
\end{proof}

\subsection{Proof of Theorem \ref{th:G-SVE_bias}}
\label{appendix:gsv_bias}

By Equation~\ref{eq:G-ACF_expec_breakdown},
%
%
\[
\mathbb{E} \left[\hat{\Upsilon}_{G}(k) \right] = \left(1 - \dfrac{\abs{k}}{n}\right)\Upsilon(k) + O_1 + O_2    \,.
\]
where both $O_1 \textrm{ and } O_2$ are the small order terms where $O_1 = n^{-1}\abs{k} \mathcal{O}(n^{-1}) \text{ and } O_2 = \mathcal{O}(n^{-1})$. By our assumptions, $\sum_{k=-\infty}^{\infty}\Upsilon(k) < \infty$.  
For a truncation point $b_n$, $w(k/b_n)=0$ for all $\abs{k} > b_n$. Therefore, SV estimator can be written as a weighted sum of estimated autocovariances from lag $-n+1$ to $n-1$. Consider the G-SV estimator,
\begin{align*}
 \mathbb{E} \left[\hat{\Sigma}_{G} - \Sigma \right] & = \sum_{k=-n+1}^{n-1} w\left(\dfrac{k}{b_n}\right)\mathbb{E} \left[\hat{\Upsilon}_{G}(k) \right] - \sum_{k=-\infty}^{\infty}\Upsilon(k)\\
    &= \sum_{k=-n+1}^{n-1}  w\left(\dfrac{k}{b_n}\right)\left[\left(1-\dfrac{k}{n}\right)\Upsilon(k) + O_1 + O_2 \right]  - \sum_{k=-\infty}^{\infty}\Upsilon(k)\\
    &= \sum_{k=-n+1}^{n-1} \left[ w\left(\dfrac{k}{b_n}\right)\left(1-\dfrac{\abs{k}}{n}\right)\Upsilon(k)\right]  - \sum_{k=-\infty}^{\infty}\Upsilon(k) + \sum_{k=-n+1}^{n-1}\left[  w\left(\dfrac{k}{b_n}\right)  (O_1 + O_2) \right] \\ 
    & = P_1 + P_2\,, \numberthis \label{eq:bias_gsve}
\end{align*}
where 
\begin{align*}
    P_1 &= \sum\limits_{k=-n+1}^{n-1}\left[w\left(\dfrac{k}{b_n}\right)\left(1-\dfrac{\abs{k}}{n}\right)\Upsilon(k) \right] - \sum\limits_{k=-\infty}^{\infty}\Upsilon(k) \textrm{ and }\\   
    P_2 &= \sum_{k=-n+1}^{n-1}\left[  w\left(\dfrac{k }{b_n}\right)\left(O_1 + O_2\right)\right]\,.
\end{align*}
Similar to \cite{hannan:1970}, we break $P_1$ into three parts. Note that notation $A = o(z)$ for matrix $A$  implies $A^{ij} = o(z)$ for every $(i,j)$th element of the matrix $A$. Consider,
\begin{equation}
\label{eq:P1_decomp}
P_1 = -\sum_{\abs{k}\geq n}\Upsilon(k)  -  \sum_{k = -n+1}^{n-1}w\left(\dfrac{k}{n}\right)\dfrac{\abs{k}}{n}\Upsilon(k)- \sum_{k = -n+1}^{n-1}\left(1-w\left(\dfrac{k}{n}\right)\right)\Upsilon(k)\,.  
\end{equation}
We deal with the three subterms of term $P_1$ individually. First,
\begin{align*}
 {-\sum_{\abs{k}\geq n}\Upsilon(k)} &\leq  \sum_{\abs{k}\geq n} \abs{\dfrac{k}{n}}^q   \Upsilon(k) = \dfrac{1}{b_n^q}\abs{\dfrac{b_n}{n}}^q \sum_{k\geq n}\abs{k}^q  \Upsilon(k) = o\left(\dfrac{1}{b_n^q}\right)\,, \numberthis \label{eq:P1_first}
\end{align*}
since $\sum_{\abs{k}\geq n}\abs{k}^q \Upsilon(k)  < \infty$. Next,
\begin{align*}
 \sum_{k = -n+1}^{n-1}w\left(\dfrac{k}{n}\right)\dfrac{\abs{k}}{n}\Upsilon(k)   &\leq \dfrac{c}{n}\sum_{k = -n+1}^{n-1}\abs{k} \Upsilon(k)  \,.
\end{align*}
For $q\geq 1$,
\begin{align*}
\dfrac{c}{n}\sum_{k = -n+1}^{n-1}\abs{k}  \Upsilon(k) &\leq \dfrac{c}{n}\sum_{k = -n+1}^{n-1}\abs{k}^q  \Upsilon(k) = \dfrac{1}{b_n^q}\dfrac{b_n^q}{n} c\sum_{k = -n+1}^{n-1}\abs{k}^q  \Upsilon(k)  = o\left(\dfrac{1}{b_n^q}\right)\,.
\end{align*}
          
For $q <1$,
\begin{align*}
   \dfrac{c}{n}\sum_{k = -n+1}^{n-1}\abs{k} \Upsilon(k)  &\leq c\sum_{k =-n+1}^{n-1}\abs{\dfrac{k}{n}}^q  \Upsilon(k)  = \dfrac{1}{b_n^q}\dfrac{b_n^q}{n^q} c \sum_{k =-n+1}^{n-1}\abs{k}^q \Upsilon(k)  = o\left(\dfrac{1}{b_n^q}\right) \, .
\end{align*}
So,
\begin{equation}
\label{eq:P1_second}
 \sum_{k = -n+1}^{n-1}w\left(\dfrac{k}{n}\right)\dfrac{|k|}{n}\Upsilon(k) = o \left(\dfrac{1}{b_n^q} \right)
\end{equation}
Lastly, by our assumptions, for $x \to 0$
\[
\dfrac{1 - w(x)}{\abs{x}^q} = k_q + o(1)\,.
\]
For $x = k/b_n$, $\abs{k/b_n}^{-q}(1-w(k/b_n))$ converges boundedly to $k_q$ for each $k$.
So,
\begin{align*}
     \sum_{k = -n+1}^{n-1}\left(1-w\left(\dfrac{k}{b_n}\right)\right)\Upsilon(k) &= -\dfrac{1}{b_n^q}\sum_{k = -n+1}^{n-1}  \left(\dfrac{|k|}{b_n}\right)^{-q} {\left(1-w\left(\dfrac{k}{b_n}\right)\right)}|k|^q \Upsilon(k) \\
     & = -\dfrac{1}{b_n^q}\sum_{k = -n+1}^{n-1}   \left[k_q + o(1) \right]|k|^q \Gamma(k) \\
     & = -\dfrac{k_q \Phi^{(q)} }{b_n^q} + o\left( \dfrac{1}{b_n^q} \right) \numberthis \label{eq:P1_third}\,.
\end{align*}

Finally, we will solve for $P_2$. Note that $O_1 = (\abs{k}/n)\mathcal{O}(1/n)$ and $O_2 = \mathcal{O}(1/n)$. So,
  \begin{align*}
    P_2  & \leq \sum_{k=-n+1}^{n-1}\abs{w\left(\dfrac{k}{b_n}\right)O_1} + \sum_{k=-n+1}^{n-1}\abs{w\left(\dfrac{k}{b_n}\right)O_2} \\
      &= \mathcal{O}\left(\dfrac{1}{n}\right)\sum_{k=-n+1}^{n-1}\dfrac{\abs{k}}{n}\abs{w\left(\dfrac{k}{b_n}\right)} + \mathcal{O}\left(\dfrac{1}{n}\right)\sum_{k=-b_n+1}^{b_n-1}\abs{w\left(\dfrac{k}{b_n}\right)}\\
      &= \mathcal{O}\left(\dfrac{b_n}{n}\right)\dfrac{1}{b_n} \sum_{k=-b_n+1}^{b_n-1}\abs{w\left(\dfrac{k}{b_n}\right)}\\
      &= \mathcal{O}\left(\dfrac{b_n}{n}\right) = o\left(\dfrac{b_n}{b_n^{q+1}}\right) = o\left(\dfrac{1}{b_n^{q}}\right) \,. \numberthis \label{eq:P_2}
 \end{align*}

Using \eqref{eq:P1_decomp},\eqref{eq:P1_first}, \eqref{eq:P1_second}, \eqref{eq:P1_third}, and \eqref{eq:P_2} in \eqref{eq:bias_gsve},
\[
\mathbb{E} \left[\hat{\Sigma}_{G} - \Sigma \right] = -\dfrac{k_q \Phi^{(q)} }{b_n^q} + o\left( \dfrac{1}{b_n^q} \right)\,,
\]
which completes the result.
%

\subsection{Proof of Theorem~\ref{th:G-SVE_variance}}
\label{appendix:variance}
 
By Lemma~\ref{lemma:G-SVE_breakdown},
\begin{align*}
 &\left| \hat{\Sigma}_{G}^{ij} - \tilde{\Sigma}^{ij} \right|\\
 & \leq \dfrac{1}{m} \sum_{s=1}^m \left| \sum_{k=-b_n+1}^{b_n-1} w \left(\dfrac{k}{b_n} \right) \sum_{t \in I_k}   \left[ \left( \dfrac{(Y_{s,t} - \mu_g)_i(\mu_g-\bar{\bar{Y}})_j}{n}\right)+ \left(\dfrac{(\mu_g-\bar{\bar{Y}})_i(Y_{s,t+k}-\mu_g)_j}{n}\right) \right] \right.\\
& \quad \quad  \left. + (\mu_g-\bar{\bar{Y}})(\mu_g-\bar{\bar{Y}})^T\sum_{k=-b_n+1}^{b_n-1}\left(\dfrac{n-|k|}{n}\right)w\left(\dfrac{k}{n}\right) \right|  < D^2g_1(n) + Dg_2(n) + g_3(n)\,,
\end{align*}
where $g_1(n), g_2(n), g_3(n) \to 0$ as $n \to \infty$ as defined in Lemma~\ref{lemma:G-SVE_breakdown}. Then there exists an $N_0$ such that
\begin{align*}
\left(\hat{\Sigma}_{G}^{ij} - \tilde{\Sigma}^{ij} \right)^2 &= \left(\hat{\Sigma}_{G}^{ij} - \tilde{\Sigma}^{ij} \right)^2 \, I(0 \leq n \leq N_0) + \left(\hat{\Sigma}_{G}^{ij} - \tilde{\Sigma}^{ij} \right)^2 \, I(n > N_0)\\
& \leq \left(\hat{\Sigma}_{G}^{ij} - \tilde{\Sigma}^{ij} \right)^2 \, I(0 \leq n \leq N_0) +  \left(D^2g_1(n) + Dg_2(n) + g_3(n) \right)^2 I(n > N_0)\\
& := g_n^*(Y_{1,1}, \dots, Y_{1n}, \dots, Y_{m1}, \dots, Y_{mn})\,.
\end{align*}
But since by assumption $\E D^4 <\infty$ and the fourth moment is finite,
\[
\E \left| g_n^* \right| \leq  \E \left[\left(\hat{\Sigma}_{G}^{ij} - \tilde{\Sigma}_{A}^{ij} \right)^2 \right] + \E \left[\left(D^2g_1(n) + Dg_2(n) + g_3(n) \right)^2 \right] < \infty\,.
\]
Thus, $\E \left| g_n^* \right| < \infty$ and further as $n \to \infty$, $g_n \to 0$ under the assumptions. Since $g_1, g_2, g_3 \to 0$, $\E g_n^* \to 0$. By the majorized convergence theorem \citep{zeid:2013}, as $n \to \infty$,
\begin{equation}
\label{eq:squared_mean_diff}
  \E \left[\left(\hat{\Sigma}_{G}^{ij} - \tilde{\Sigma}^{ij} \right)^2 \right] \to 0\,.
\end{equation}
We will use \eqref{eq:squared_mean_diff} to show that the variances are equivalent. Define,
\[
\xi\left(\hat{\Sigma}_{G}^{ij}, \tilde{\Sigma}^{ij} \right) = \Var\left(\hat{\Sigma}_{G}^{ij} - \tilde{\Sigma}^{ij} \right) + 2 \E\left[ \left(\hat{\Sigma}_{G}^{ij} -  \tilde{\Sigma}^{ij} \right) \left(\tilde{\Sigma}^{ij}  - \E \left( \tilde{\Sigma}^{ij} \right) \right) \right]\,.
\]
We will show that the above is $o(1)$. Using Cauchy-Schwarz inequality followed by \eqref{eq:squared_mean_diff},
\begin{align*}
\left|  \xi\left(\hat{\Sigma}_{G}^{ij}, \tilde{\Sigma}^{ij} \right) \right| & \leq \left| \Var\left(\hat{\Sigma}_{G}^{ij} -  \tilde{\Sigma}^{ij} \right) \right| + \left| 2 \E\left[ \left(\hat{\Sigma}_{G}^{ij} - \tilde{\Sigma}^{ij} \right) \left(\tilde{\Sigma}^{ij}  - \E \left( \tilde{\Sigma}^{ij} \right) \right) \right]\right| \\ 
& \leq \E\left[\left(\hat{\Sigma}_{G}^{ij} -  \tilde{\Sigma}^{ij} \right)^2 \right] + 2 \left| \left(\E\left[ \left(\hat{\Sigma}_{G}^{ij} - \tilde{\Sigma}^{ij} \right)^2 \right]  \Var\left(\tilde{\Sigma}^{ij}  \right)   \right)^{1/2}\right| \\ 
& = o(1) + 2\left(o(1) \left(O\left( \dfrac{b_n}{n}\right)  + o\left( \dfrac{b_n}{n}\right) \right)  \right) = o(1)\,.
\end{align*}
Finally,
\begin{align*}
 \Var\left(\hat{\Sigma}_{G}^{ij} \right)  & = \E \left[ \left(\hat{\Sigma}_{G}^{ij}  - \E \left[\hat{\Sigma}_{R}^{ij}  \right] \right)^2 \right]\\
& = \E \left[ \left(\hat{\Sigma}_{G}^{ij} \pm \tilde{\Sigma}^{ij} \pm \E \left[ \tilde{\Sigma}^{ij}\right] - \E \left[\hat{\Sigma}_{G}^{ij}  \right] \right)^2 \right]\\
& = \E\left[ \left( \left(\hat{\Sigma}_{G}^{ij} - \tilde{\Sigma}^{ij} \right) + \left(\tilde{\Sigma}^{ij}  - \E\left[\tilde{\Sigma}^{ij}\right]\right) + \left(\E\left[\tilde{\Sigma}^{ij}\right] - \E \left[\hat{\Sigma}_{G}^{ij}  \right] \right)  \right)^2 \right] \\ 
& =  \E\left[ \left(\tilde{\Sigma}^{ij}  - \E\left[\tilde{\Sigma}^{ij}\right]\right)^2 \right] + \E \left[ \left(\left(\hat{\Sigma}_{G}^{ij} - \tilde{\Sigma}^{ij} \right) + \left(\E\left[\tilde{\Sigma}^{ij}\right] - \E \left[\hat{\Sigma}_{G}^{ij}  \right] \right) \right)^2 \right] \\
& \quad \quad + 2\E\left[\left(\tilde{\Sigma}^{ij}  - \E\left[\tilde{\Sigma}^{ij}\right]\right) \left(\hat{\Sigma}_{G}^{ij} - \tilde{\Sigma}^{ij} \right) + 2 \left(\tilde{\Sigma}^{ij}  - \E\left[\tilde{\Sigma}^{ij}\right]\right) \left(\E \left[\tilde{\Sigma}^{ij}\right] - \E \left[\hat{\Sigma}_{G}^{ij}  \right] \right)\right]\\
& = \Var\left( \tilde{\Sigma}^{ij}\right) + \Var\left(\hat{\Sigma}_{G}^{ij} - \tilde{\Sigma}^{ij} \right) + 2 \E\left[ \left(\hat{\Sigma}_{G}^{ij} -  \tilde{\Sigma}^{ij} \right) \left(\tilde{\Sigma}^{ij}  - \E \left( \tilde{\Sigma}^{ij} \right) \right) \right] + o(1)\\
& = \Var\left( \tilde{\Sigma}^{ij}\right) + o(1)\,.
\end{align*}

\cite{hannan:1970} has given the calculations for variance of $\tilde{\Sigma}$ as 
\begin{equation} \label{eq:hannan_var}
\dfrac{n}{b_n}\Var(\tilde{\Sigma}^{ij}) = \left[\Sigma^{ii}\Sigma^{jj} + \left(\Sigma^{ij} \right) \right]\int_{-\infty}^{\infty}w^2(x)dx + o(1)\,.
\end{equation}
Plugging (\ref{eq:hannan_var}) into variance of $\hat{\Sigma}_{G}$ gives the result of the theorem.
 

\section{Additional Examples}
We present five additional experimental studies to illustrate the advantage of the G-ACF estimator. 
\subsection{Bayesian Poisson Change Point Model}


Consider the militarized interstate dispute (MID) data of \cite{martin2011mcmcpack} which describes the annual number of military conflicts in the United States. In order to detect the number and timings of the cyclic phases in international conflicts, we fit a Bayesian Poisson change-point model. Following \cite{martin2011mcmcpack}, we will use \texttt{MCMCpoissonChange} from \texttt{MCMCpack} to fit the model with six change-points which samples the latent states based on the algorithm in \cite{chib1998estimation}. The Poisson change-point model in \texttt{MCMCpoissonChange} is:

%
\begin{align*}
    y_t &\sim \text{Poisson}(\lambda_i), \qquad i = 1, ..., k\\
    \lambda_i &\sim \text{Gamma}(c_o, d_o), \qquad i = 1,..., k\\
    p_{ii} &\sim \text{Beta}(\alpha, \beta), \qquad i = 1,..., k\,.
\end{align*}

This yields a 7-dimensional posterior distribution in $\lambda = (\lambda_1, \dots, \lambda_7)^T$. Figure~\ref{fig:poisson-trace} shows the evolution of two chains started from random points for the second and third component. We report the ACF plots for the second component component only, however, similar behavior is observed in ACF plots of other components as well.

\begin{figure}[h]
    \centering
    \includegraphics[width = 0.7\textwidth]{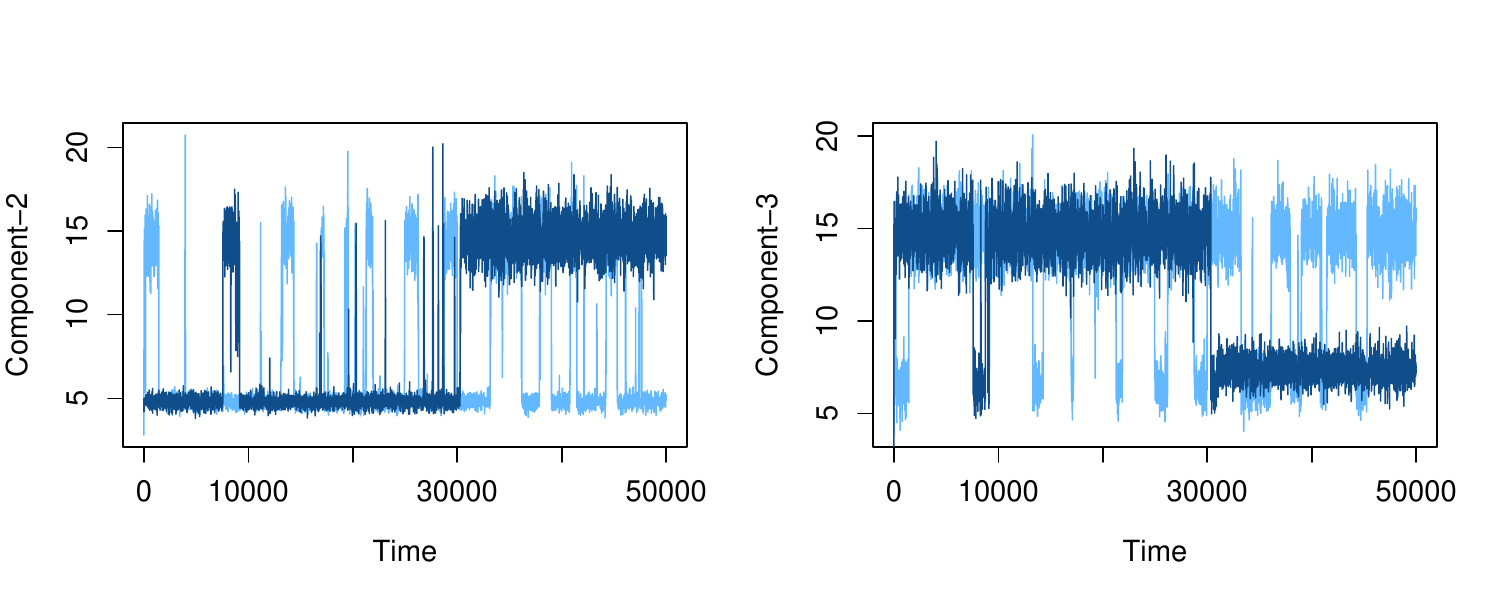}
    \caption{Change point: Trace plots for second (left) and third (right) component.}
    \label{fig:poisson-trace}
\end{figure}
\begin{figure}[h]
    \centering
      \includegraphics[width = .6\textwidth]{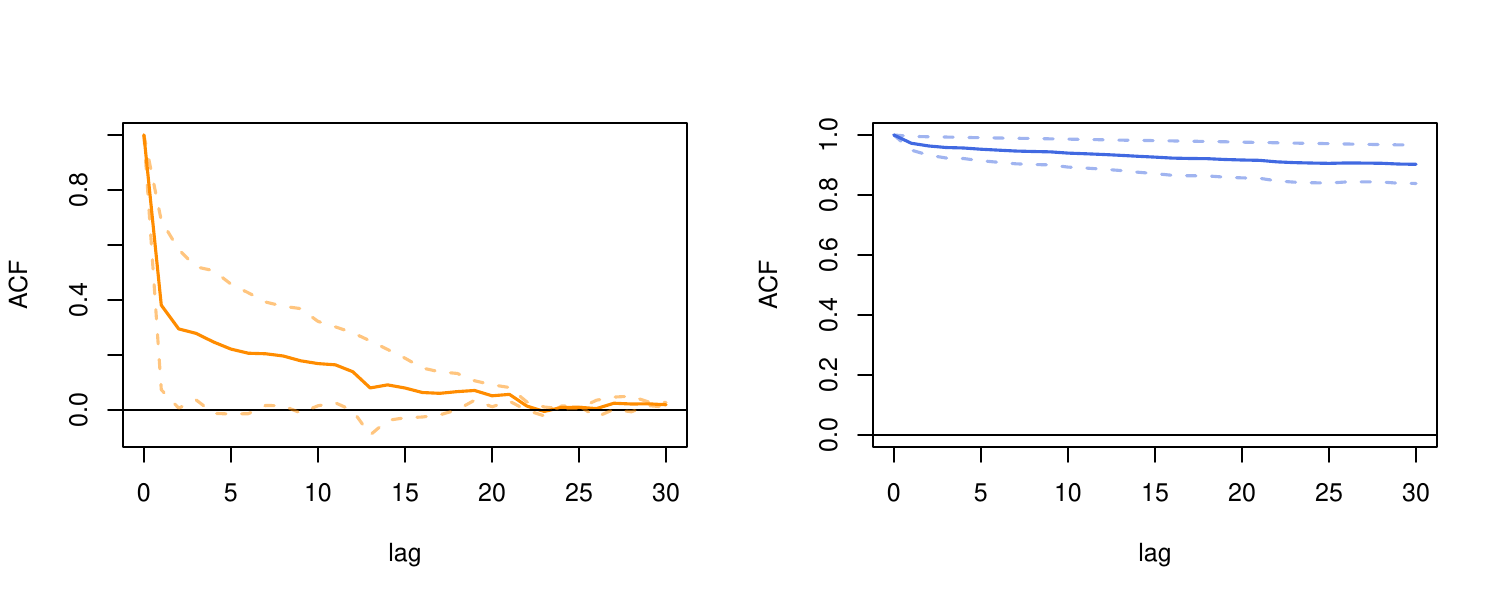}  \\ \vspace{-.5cm}
      \includegraphics[width = .6\textwidth]{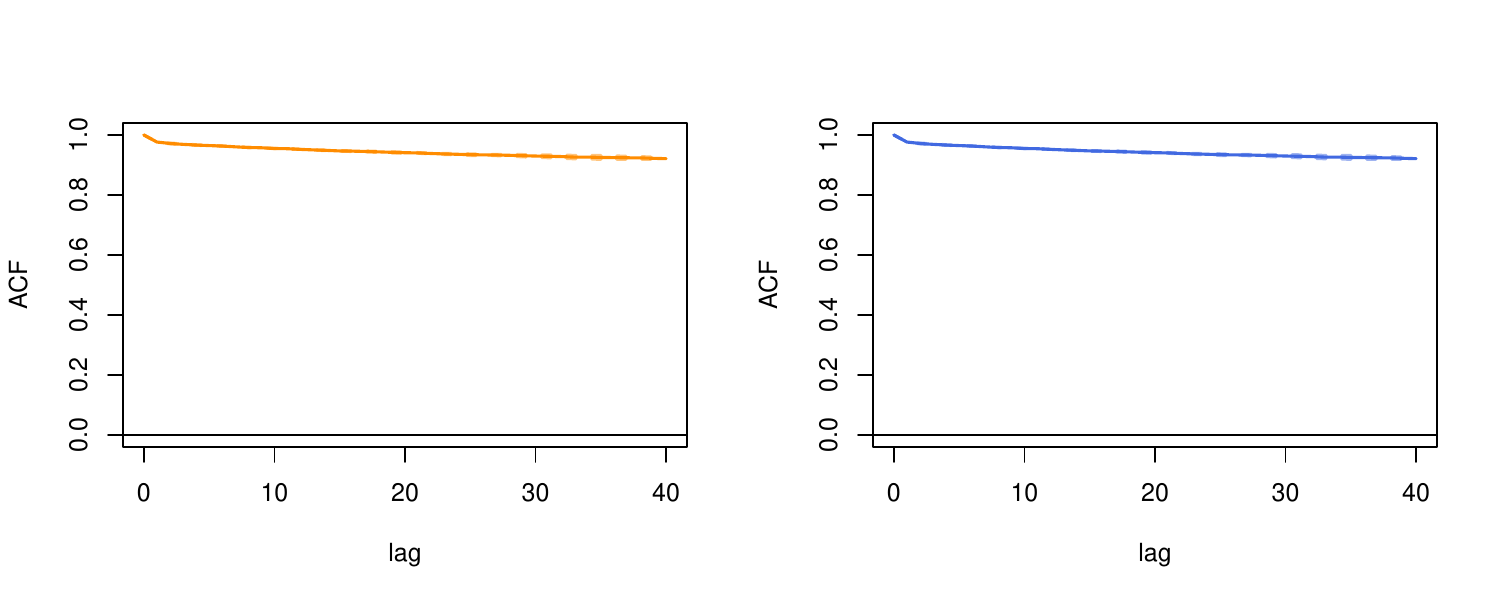}      
    \caption{Change point: ACF plots using local (left) and global (right) centering for $m=2$ parallel Markov chains. The individual chains are shown through dashed lines and average over $m$ chains is the solid line. The chain lengths are $n = 1000$ (top row) and $n = 10000$ (bottom row).}
    \label{fig:poisson-acf}
\end{figure}

Figure~\ref{fig:poisson-acf} demonstrates a striking advantage of G-ACF against locally-centered ACF in estimating the autocorrelations. For a chain length of $n=1000$, locally-centered ACF gives a false sense of security about the Markov chains whereas in reality, the process is highly correlated. 


\subsection{Network crawling}


The \texttt{faux.magnolia.high} dataset available in the \texttt{ergm} \texttt{R} package represents a simulated friendship network based on Ad-Health data (\cite{resnick1997protecting}). The school communities represented by the network data are located in the southern United States. Each node represents a student and each edge represents a friendship between the nodes it connects.

The goal is to draw nodes uniformly from the network by using a network crawler. \cite{nilakanta2019ensuring} modified the data by removing 1,022 out of 1,461 nodes to obtain a well-connected graph. This resulting social network has 439 nodes and 573 edges. We use a Metropolis-Hastings algorithm with a simple random-walk  proposal suggested by \cite{gjoka2011practical}. Each node is associated with five features namely - degree of connection, cluster coefficient, grade, binary sex indicator (1 for female, 0 for male), and binary race indicator (1 for white, 0 for others). We sample two parallel Markov chains starting from two students belonging to different races.

\begin{figure}[h]
    \centering
      \includegraphics[width = .8\textwidth]{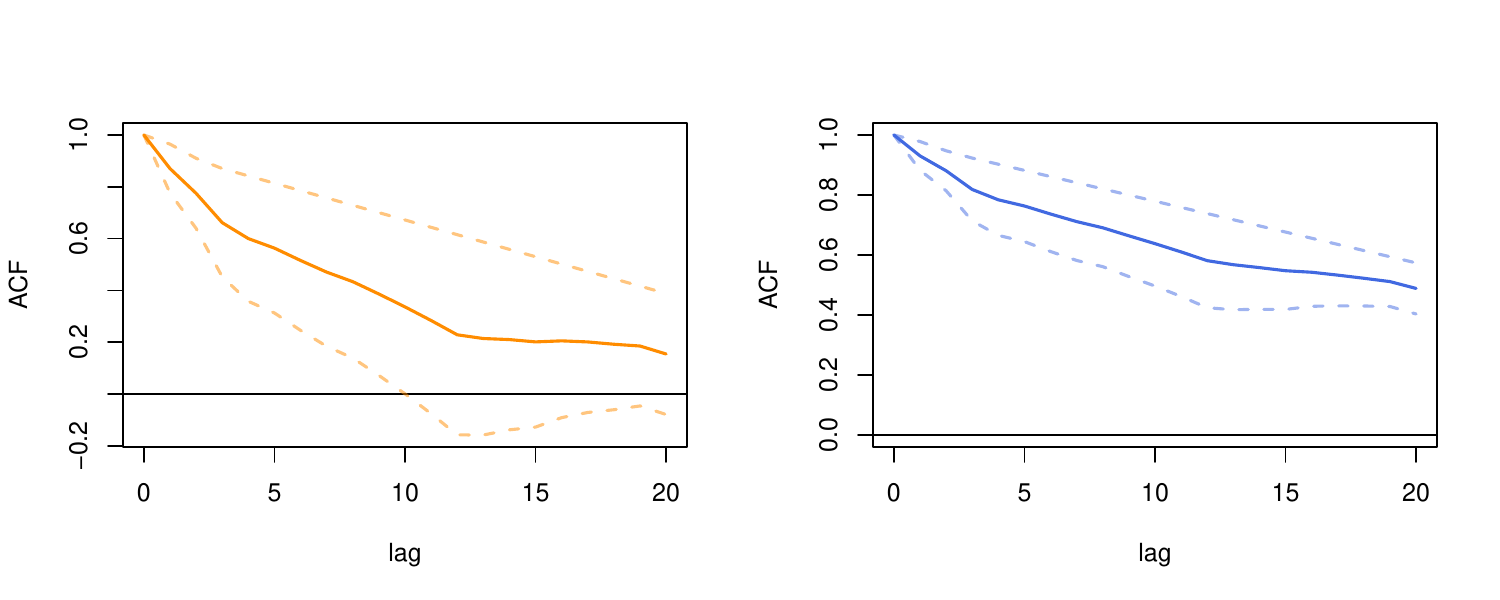}\\ \vspace{-.5cm}
      \includegraphics[width = .8\textwidth]{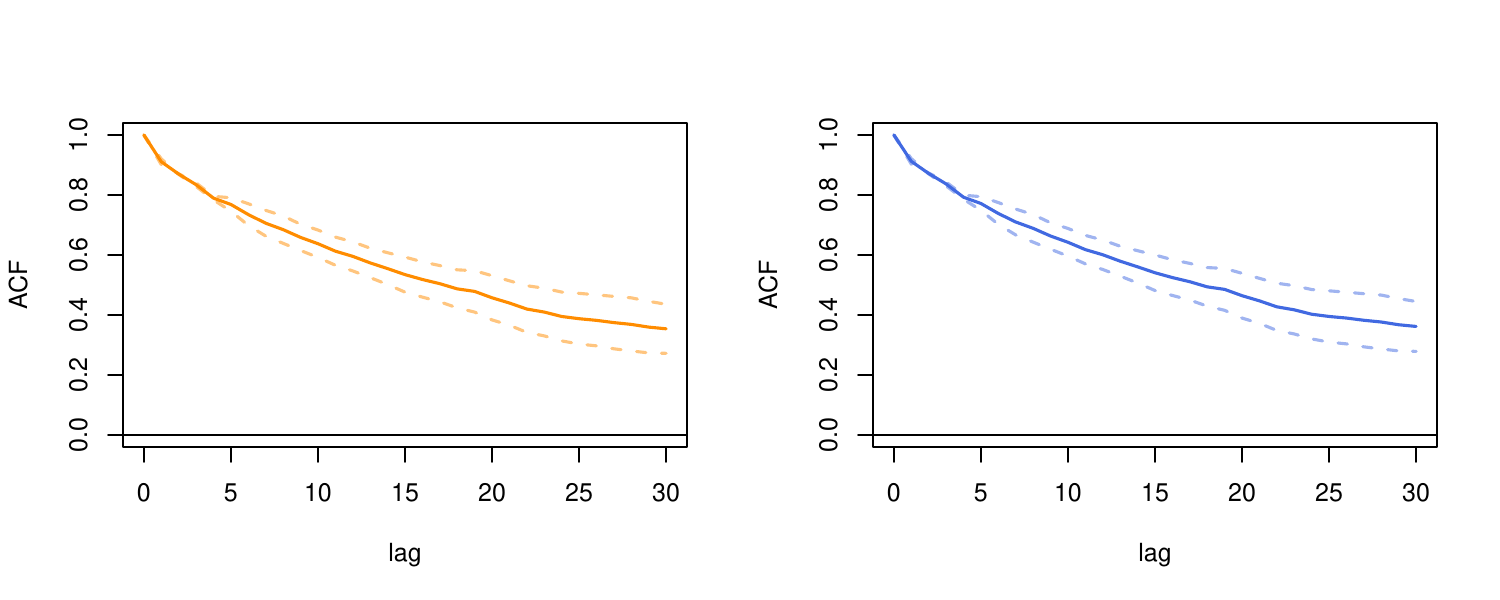}

    \caption{Crawling:  ACF plots using local-centering (left) and global-centering (right) for $m=2$ parallel Markov chains. The individual chains are shown through dashed lines and average over $m$ chains is the solid line. The chain lengths are $n = 100$ (top row) and $n = 1000$ (bottom row).}
    \label{fig:magnolia-acf}
\end{figure}
Figure~\ref{fig:magnolia-acf} shows the ACF plots for the \textit{grade} feature at two different simulation sizes for $m=2$ parallel Markov chains. For a shorter chain length, both the chains have explored different clusters and as a consequence, the local and global mean do not agree. Regardless, G-ACF displays a clear advantage over locally-centered ACF for a chain length of $n=100$ samples whereas the latter takes $n=1000$ samples to reach the truth. 

\subsection{VAR with negative autocorrelations}

In this example, we assess the performance of locally and globally-centered ACF estimators in the presence of negative autocorrelations. Consider the VAR(1) process from Example~\ref{ex:var} with $\Omega$ denoting the AR correlation matrix, $\Phi$ denoting the coefficient matrix, and the  $N(0, \Psi)$ being invariant distribution of Markov chain. We fix $\Omega$ with parameter $-0.5$. $\Xi$ has eigenvalues $-0.9$ and $-0.1$ and its diagonal entries are negative. Since $\Upsilon(k) = \Gamma(k) = \Xi^k\Psi$, autocoaviances for each component are positive when $k$ is even and negative when $k$ is odd. 

Figure~\ref{fig:neg_var-acf} shows the estimated ACF plots for the first component of the first chain against the truth in red dots. The top row shows results for $n=1000$ and bottom for $n=10000$. A negative lag-$k$ correlation emerges out of a systematic jumping of Markov chains away from its. This behavior itself ensures that the Markov chains mix well. This is in direct contrast to the positive autocorrelation case where the chains cannot explore the state space well. As a consequence, the locally-centered ACF estimator does not severely underestimate the truth here, rendering a very similar estimation quality to G-ACF, even for small $n$ case.

\begin{figure}[h]
    \centering
      \includegraphics[width = .8\textwidth]{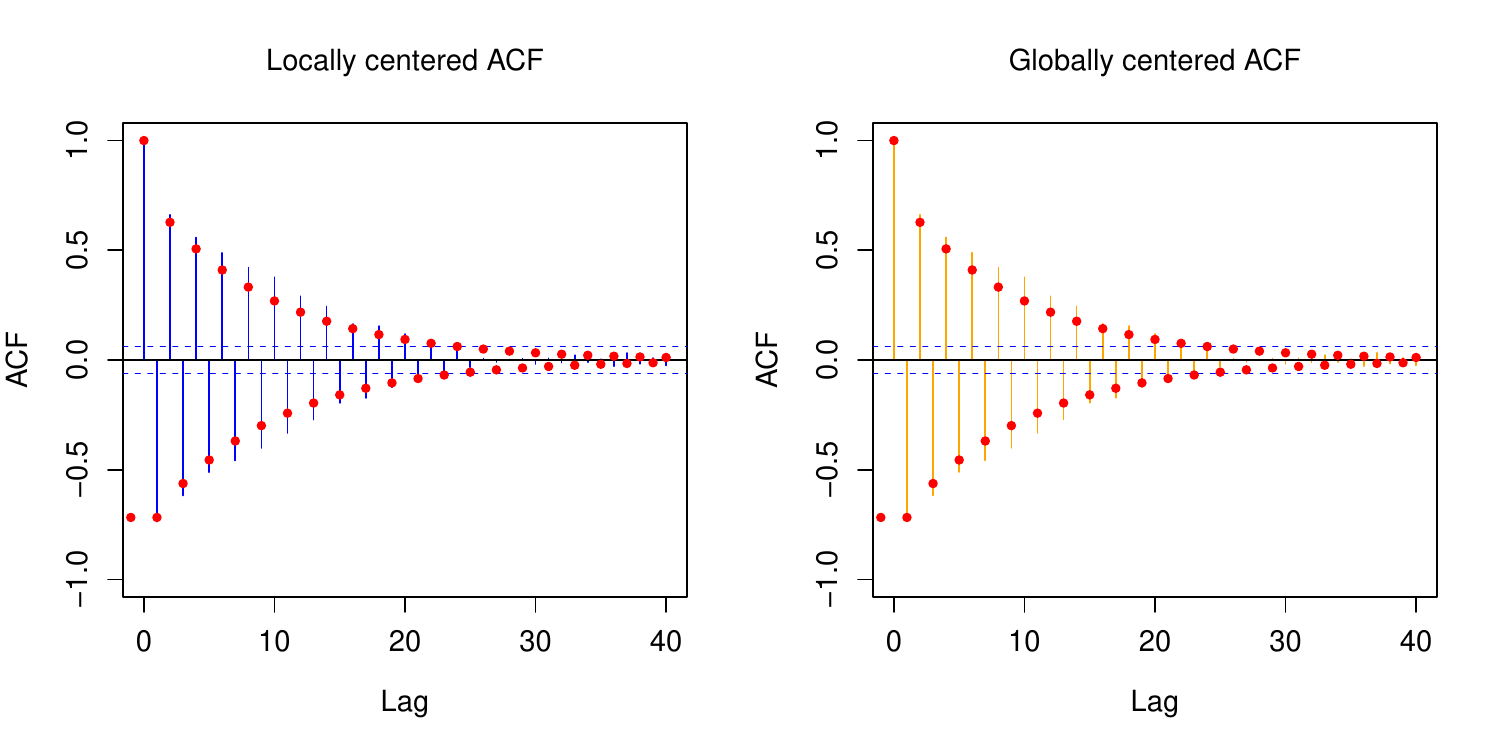}\\ \vspace{-.5cm}
      \includegraphics[width = .8\textwidth]{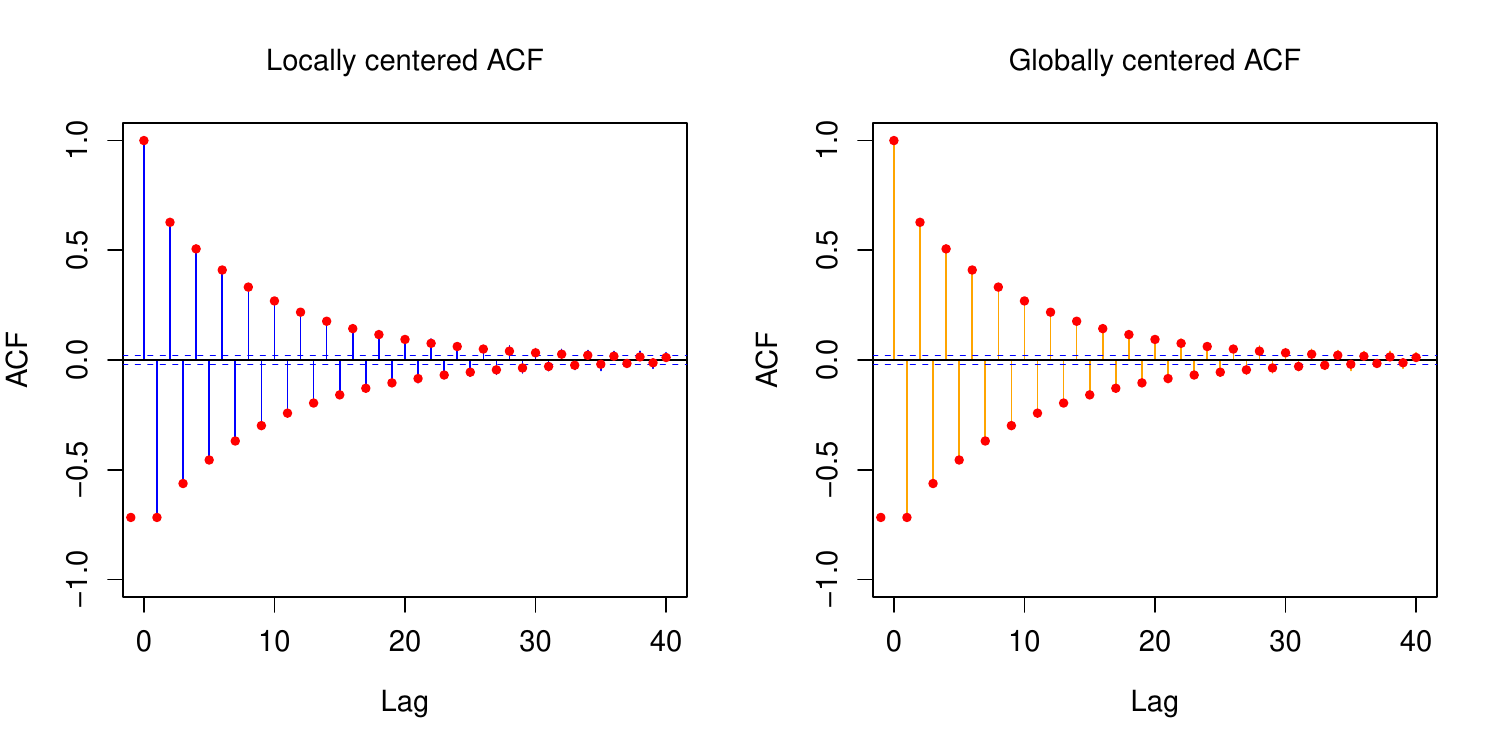}

    \caption{Negative ACF VAR:  ACF plots using local-centering (left) and global-centering (right) for $m=5$ parallel Markov chains. (Top) $n = 10^3$ and (bottom) $n = 10^4$. The red dots are the true ACF.}
    \label{fig:neg_var-acf}
\end{figure}

\subsection{High dimensional VAR}

Estimating the limiting covariance matrix in an MCMC algorithm is generally a challenging task and an area of active research. Although significant contributions have been made to this cause, the performance of such estimators in high-dimensions is not yet satisfactory, particularly for slow-mixing Markov chains.
%
  In general, the G-SV estimator will be better than the A-SV estimator, but for slow-mixing Markov chains, both estimators do not perform well at reasonable sample sizes.

  
  To investigate the relative performance of A-SV and G-SV estimator under different mixing characteristics, we implement a 100 dimensional VAR example with $m = 5$. For $p = 100$, we consider two settings: low and high target correlations that furnishes fast and slow-mixing MCMC chains, respectively. In Figure~\ref{fig:highdimvar-frob}, we plot the relative Frobenius norm of A-SV $(\|\hat{\Sigma}_A\|_F/\|\Sigma\|_F)$ and G-SV $(\|\hat{\Sigma}_G\|_F/\|\Sigma\|_F)$ estimators as a function of sample size $n$. The left plot corresponds to low target correlation that guarantees that all five chains mix well. As a consequence, both A-SV and G-SV provide similarly good estimation quality. However, for slowly-mixing chains on the right, both the A-SV and G-SV estimators are unable to estimate the truth well even for $n = 5\times10^4$. Note that the G-SV estimator provides better estimation than A-SV estimator for smaller $n$.
  
  
  \begin{figure}[h]
    \centering
    \includegraphics[width = .45\textwidth]{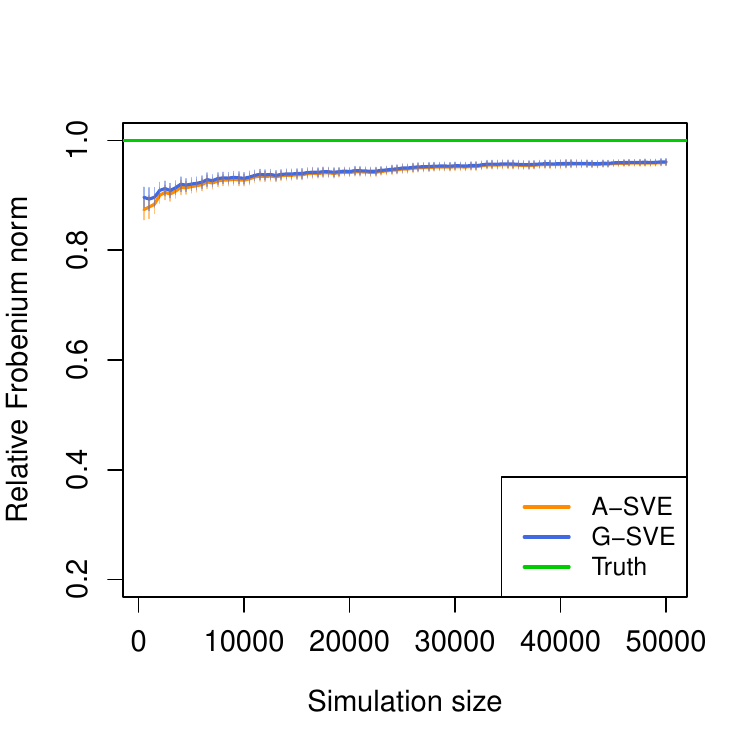}
    \includegraphics[width = .45\textwidth]{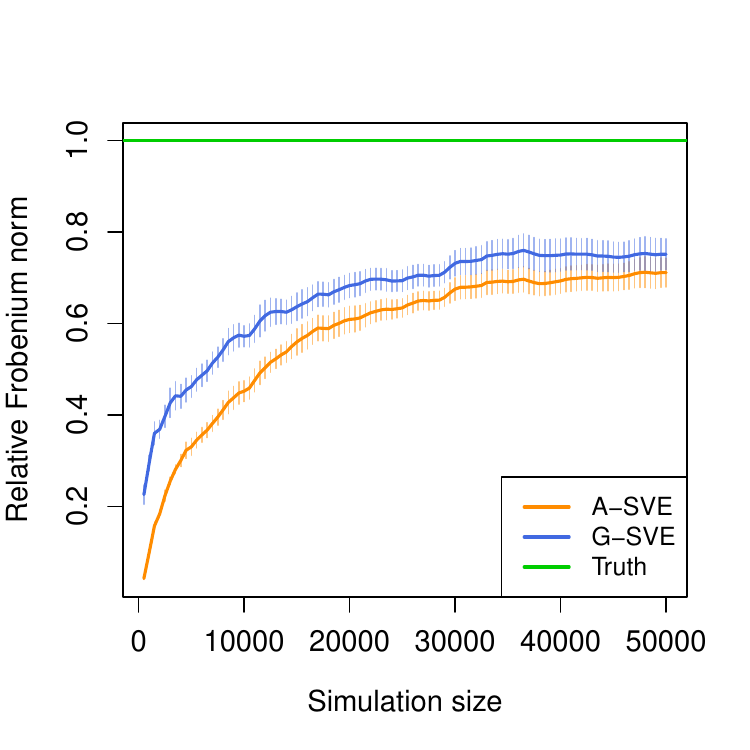}
    
    \caption{High dimensional VAR:  (Left) Running plot for relative Frobenius norm of A-SV and G-SV estimator for low target correlation. (Right) Running plot for relative Frobenius norm of A-SV and G-SV estimator for high target correlation.}
    \label{fig:highdimvar-frob}
  \end{figure}

\subsection{VAR with large $m$}

We present the VAR example with the autocorrelations estimated as a function of $m$. For the VAR model in Section~\ref{ex:var}, we study the estimated ACF as a function of $m$. We set $m = \{2,3,4, \dots, 100\}$ and compare the  autocorrelation at lag 1 and lag 40 for the first component, using both globally-centered and locally-centered ACFs. 

\begin{figure}[h]
\centering
   \includegraphics[width=.40\linewidth]{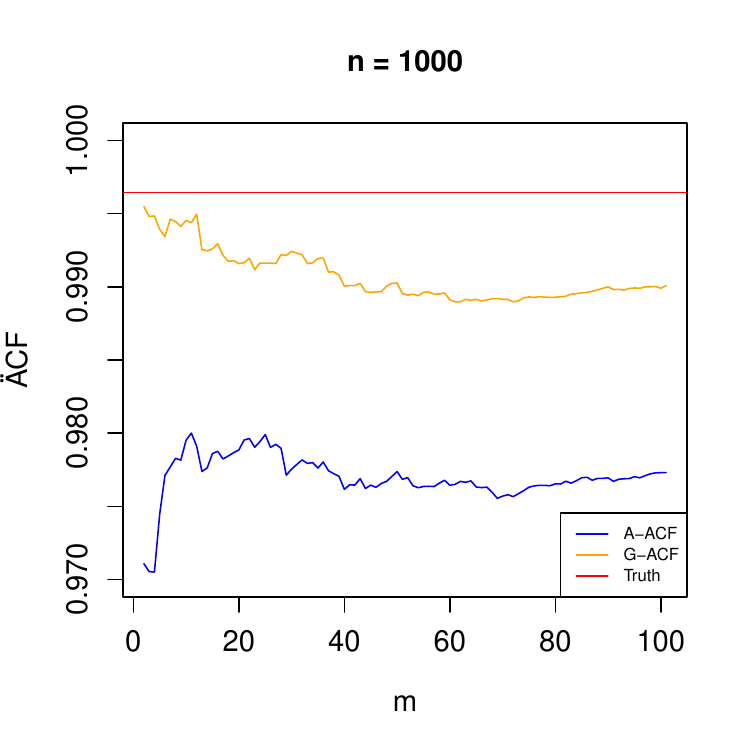}
   \includegraphics[width=.40\linewidth]{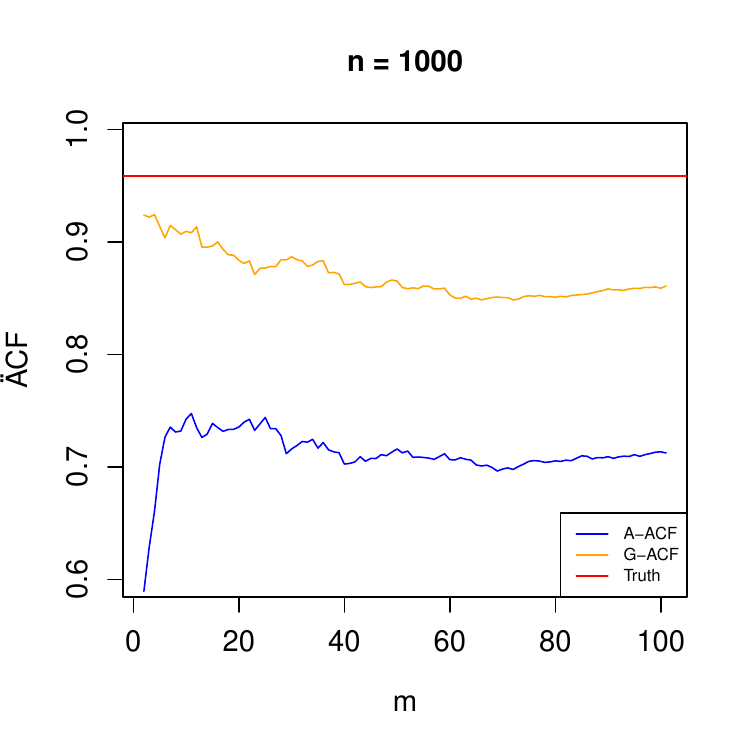}
  
    \includegraphics[width=.4\linewidth]{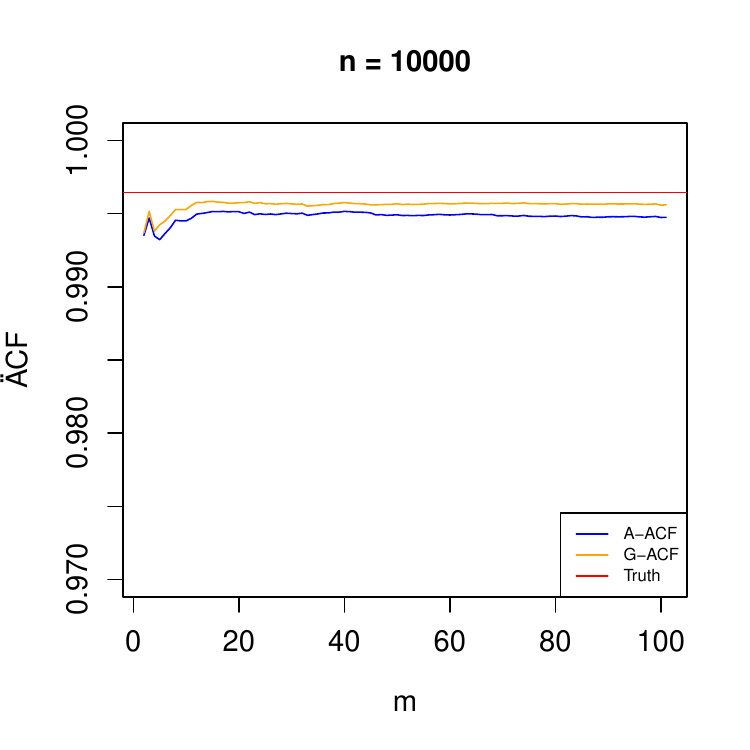}
  \includegraphics[width=.4\linewidth]{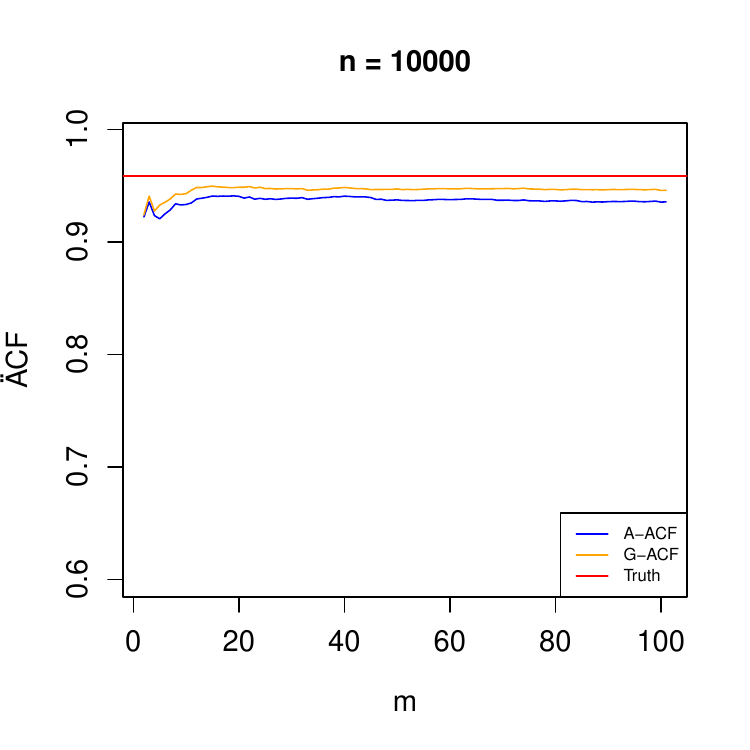}
\caption{VAR. Estimated ACF at lag 1 (left) and lag 40 (right) for the first component of the VAR, for increasing $m$. (Top) $n = 10^3$  and (bottom) $n =  10^4$. The red lines are the true autocorrelation.}
\label{fig:var-acf-mchains}
\end{figure}

Figure~\ref{fig:var-acf-mchains} presents these estimated autocorrelations for sample sizes $10^3$ and $10^4$. Naturally, as $m$ increases, the variability reduces. However, we note that the deviation from the truth remains more or less the same as a function of $m$. Globally-centered autocorrelations consistently perform better than locally-centered autocorrelations; the difference between them does not reduce  All chains here are started from stationarity.

\bibliographystyle{apalike}
\bibliography{sample}
\end{document}